\PassOptionsToPackage{nosumlimits,nonamelimits}{amsmath}
\documentclass[acmsmall,nonacm,screen,final]{acmart}

\renewcommand\footnotetextcopyrightpermission[1]{} %

\usepackage[english]{babel}

\addto\extrasenglish{%
}

\makeatletter
\providecommand\@ACM@color@native{true}
\makeatother

\hypersetup{
  colorlinks=true,
  linkcolor=blue,
  citecolor=blue,
  urlcolor=blue
}

\AtEndPreamble{%
	\theoremstyle{acmdefinition}
	\newtheorem{remark}[theorem]{Remark}
	\newtheorem{assumption}[theorem]{Assumption}
}

\usepackage{microtype}

\usepackage{thmtools} %

\usepackage{color}[final]

\usepackage{etoolbox} %

\usepackage{color}

\usepackage{url} 

\usepackage{etoolbox}
\usepackage{needspace}

\usepackage{hypcap}
\usepackage{accents}

\usepackage{wasysym} %

\usepackage{tikz-cd}

\tikzset{
  commutative diagrams/.cd,
  arrow style = tikz,
  diagrams    = {>=stealth},
  row sep     = large, 
  column sep  = huge
}

\tikzset{
    cong/.style={draw=none,edge node={node [sloped, allow upside down, auto=false]{$\cong$}}}
, iso/.style={draw=none,every to/.append style={edge node={node [sloped, allow upside down, auto=false]{$\cong$}}}}
}

\usetikzlibrary{calc}

\usepackage{textcomp}

\usepackage{enumitem} 
\setlist[itemize]{label={\small$\bullet$}}

\BeforeBeginEnvironment{align}{\noindent\ignorespaces}

\usepackage{proof}
\usepackage{xspace}

\input{catprog}
\renewcommand{\by}[1]{\text{/$\mspace{-2mu}$/~#1}} 

\usepackage{ifdraft} 
\ifdraft{
  \usepackage{showlabels} 
  
  \usepackage[layout=footnote,draft]{fixme}

}{
 \usepackage[layout=footnote,final]{fixme}
}

\FXRegisterAuthor{sg}{asg}{SG}	%
\FXRegisterAuthor{pp}{app}{PP}	%

\usepackage{adjustbox}
\newcommand{\val}{\mathsf{v}}
\newcommand{\com}{\mathsf{c}}
\newcommand{\nats}{\mathbb{N}}

\newcommand{\mS}{\mu\Sigma}

\newcommand{\init}{\oname{init}}

\newcommand{\dl}{\chi}

\newcommand{\lcomp}{\mathbin{\raisebox{1pt}{\scalebox{.6}{$\LEFTcircle$}}}}
\newcommand{\fcomp}{\mathbin{\raisebox{1pt}{\scalebox{.6}{$\CIRCLE$}}}}%

\newcommand{\myparagraph}[1]{\medskip\noindent{\bfseries\sffamily #1.}}
\renewcommand{\comp}{\cdot}

\newcommand{\Fst}{\Pi_1} %
\newcommand{\Snd}{\Pi_2} %
\newcommand{\klstar}{\sharp}  			%
\newcommand{\istar}{\dagger}  	
\newcommand{\mmul}{\mu}  	
		
\newcommand{\mSv}{\mS_\val}
\newcommand{\mSc}{\mS_\com}

\newcommand{\Sigmas}{\Sigma^\star}

\newcommand{\ar}{\oname{ar}}
\newcommand{\lar}{\oname{ar_l}}
\newcommand{\sar}{\oname{ar_s}}

\newcommand{\klplus}{\boxplus}

\makeatletter
\newsavebox{\@brx}
\newcommand{\llangle}[1][]{\savebox{\@brx}{\(\m@th{#1\langle}\)}%
  \mathopen{\copy\@brx\kern-0.5\wd\@brx\usebox{\@brx}}}
\newcommand{\rrangle}[1][]{\savebox{\@brx}{\(\m@th{#1\rangle}\)}%
  \mathclose{\copy\@brx\kern-0.5\wd\@brx\usebox{\@brx}}}
\makeatother

\newcommand{\app}{\,}
\usepackage{proof}

\newcommand{\fset}{{\mathbb{F}}}
\newcommand{\Nat}{{\mathbb{N}}}

\newcommand{\exend}{\hfill{\rotatebox[origin=c]{45}{$\Box$}}}

\makeatletter
\let\dir@frac\frac
\newcommand{\inv@frac}[2]{\dir@frac{#2}{#1}}

\renewcommand{\frac}{\@ifstar{\inv@frac}{\dir@frac}}
\makeatother

\newcommand{\mon}{\bullet}
\newcommand{\monto}{\mathrel{-\mkern-1mu}\joinrel\mathrel{\bullet}} %
\newcommand{\Pt}{V}

\newcommand{\xCL}{\textbf{xCL}\xspace}  
\newcommand{\xTCL}{\textbf{xTCL}\xspace}  	
\newcommand{\Ty}{\mathsf{Ty}\xspace}

\newcommand{\arty}[2]{#1 \rightarrowtriangle #2}

\newcommand{\tcomp}{\textsf{app}}

\newcommand{\true}{\mathsf{true}}
\newcommand{\false}{\mathsf{false}}	
\newcommand{\bool}{\mathsf{bool}}	
\newcommand{\fpc}{\mathsf{fix}}	
\newcommand{\casec}{\mathsf{if}}
\newcommand{\ndet}{\oplus} %
\newcommand{\paral}{\parallel} %
\newcommand{\amb}{\mathsf{amb}}

\makeatletter
\newcommand{\superimpose}[2]{%
  {\ooalign{$#1\@firstoftwo#2$\cr\hfil$#1\@secondoftwo#2$\hfil\cr}}
}
\makeatother

\newcommand{\paralv}{\mathrel{%
    \mathchoice{\PARV}{\PARV}{\scriptsize\PARV}{\tiny\PARV}
}}

\def\PARV{{%
    \setbox0\hbox{$\parallel$}%
    \rlap{\hbox to \wd0{\hss\rule[-.5ex]{.75em}{.12ex}\hss}}\box0
}}

\usepackage[final]{listings}

\makeatletter
\def\todo#1{} %

\ifdraft
{
  \def\ensuretext#1{\ifmmode\text{#1}\else{#1}\fi}
  \def\todo{\@ifstar\@@todo\@todo}
  \ifpdf
    \long\def\@todo#1{\ensuretext{\ttfamily\color{ACMRed}TODO: #1}}
    \long\def\@@todo#1{{\hfill\\\noindent\ttfamily\color{ACMRed}TODO: #1\hfill\par}}
  \else
  \fi
}
{}
\makeatother

\renewcommand{\xto}[1]{\mathrel{\raisebox{-1pt}{$\;\xrightarrow{\raisebox{-1.5pt}[0pt][0pt]{ \ensuremath{ \scriptstyle{#1}}}}$}}}

\begin{document}\allowdisplaybreaks
\let\cedilla\c
\renewcommand{\c}{\colon}

\title{Big Steps in Higher-Order Mathematical Operational Semantics}

\author{Sergey Goncharov}
\orcid{0000-0001-6924-8766}             %
\affiliation{
  \department{School of Computer Science}              %
  \institution{University of Birmingham} %
  \city{Birmingham}
  \country{UK}                    %
}
\email{s.goncharov@bham.ac.uk}          %

\author{Pouya Partow}
\orcid{0000-0001-6924-8766}             %
\affiliation{
  \department{School of Computer Science}              %
  \institution{University of Birmingham} %
  \city{Birmingham}
  \country{UK}                    %
}
\email{p.partow@bham.ac.uk}          %

\author{Stelios Tsampas}
\orcid{0000-0001-8981-2328}             %
\affiliation{
  \institution{University of Southern Denmark (SDU)}            %
  \city{Odense}
  \country{Denmark}                    %
}
\email{stelios@imada.sdu.dk}          %

\begin{abstract}
\emph{Small-step} and \emph{big-step operational semantics} are two fundamental 
styles of structural operational semantics (SOS), extensively used in practice. The former
one is more fine-grained and is usually regarded as primitive, as it only
defines a one-step reduction relation between a given program and its direct
descendant under an ambient \emph{evaluation strategy}. The latter one implements, 
in a self-contained manner, such a strategy 
directly by relating a program to the net result of the evaluation process. The agreement between these two styles of semantics is one of the key
pillars in operational reasoning on programs; however,
such agreement is typically proven from scratch every time on a
case-by-case basis. A general, abstract mathematical argument behind this agreement is up till now missing. We cope with this issue within the framework of \emph{higher-order
mathematical operational semantics} by providing an abstract categorical notion
of big-step SOS, complementing the existing notion of abstract higher-order
GSOS. Moreover, we introduce a general construction for deriving the former from
the latter, and prove an abstract equivalence result between the two.
\end{abstract}

	


\maketitle

\section{Introduction}%

Operational semantics of programming languages comes in two major styles: 
the \emph{small-step} and the \emph{big-step}. In both cases
we deal with a rule-based specification of program behaviour, however, the 
respective rules operate with two principally different judgement formats. In paradigmatic 
cases, such as the (call-by-name) $\lambda$-calculus, the small-step judgements
have the form $t\to t'$, and the big-step judgements have the form $t\Dar v$. Here,~$t$ 
is a (closed) program,~$t'$ is its direct successor under the reduction relation of interest,
and $v$ is a final \emph{value}, to which~$t$ evaluates. The desired 
connection between these two judgements is expressed by the fundamental equivalence:
\begin{align}\label{eq:sem-eq}\tag{$\star$}
t\Dar v \iff t\to^\star v\land v\dar
\end{align}
where $\to^\star$ is the transitive-reflexive closure of $\to$ and $v\dar$ means 
that $v\to v'$ for no~$v'$. In the case of the call-by-name $\lambda$-calculus, the 
setup is particularly simple: there are precisely two small-step rules and two
big-step rules -- see \autoref{fig:cbn-lambda}. 
Even in this simple case, proving the equivalence~\eqref{eq:sem-eq} is non-trivial and requires some creativity.
  
\begin{figure}[t]
\begin{align*}%
\infer{(\lambda x.\, t)s\to t[s/x]}{} && 
\infer{ts\to t's}{t\to t'} &\qquad& 
\infer{\lambda x.\,t\Dar \lambda x.\,t}{} && 
\infer{ts\Dar v}{t\Dar\lambda x.\,t'\quad t'[s/x]\Dar v} 
\end{align*}%
\caption{Operational semantics of call-by-name $\lambda$-calculus (small-step and big-step).}
\label{fig:cbn-lambda}
\end{figure}

The equivalence~\eqref{eq:sem-eq} plays an important role in various settings where both small-step and big-step semantics are defined. These settings can vary widely, encompassing different evaluation strategies, language features, computational effects (such as partiality, nondeterminism, and state), and even extensions to quantitative semantics~\cite{LagoZorzi12,Laird16}. While in many standard cases the proof of~\eqref{eq:sem-eq} follows familiar and often routine patterns it can still become tedious and error-prone when applied to expressive, feature-rich languages. In the absence of a unified mathematical framework that abstracts these patterns, such proofs need to be re-established manually, which may obscure reuse and increase the risk of oversight, unless the proofs are fully formalized and machine-checked.

The main goal of our present work is to provide suitable abstractions for the notions of small-step and big-step operational semantics, enabling us to formulate and prove~\eqref{eq:sem-eq} at a high level of generality, particularly by parametrizing over suitable notions of \emph{syntax} and \emph{behaviour}.
To that end, we capitalize on recent advances in \emph{higher-order mathematical 
operational semantics}~\cite{GoncharovMiliusEtAl23}, which is a higher-order extension
of Turi and Plotkin's \cite{TuriPlotkin97} (first-order) mathematical operational semantics.
Thus, a general notion of (higher-order) small-step semantics in the form of an \emph{abstract HO-GSOS law} is already available from previous work~\cite{GoncharovMiliusEtAl23}.
However, motivated by~\eqref{eq:sem-eq}, in this paper, we develop a refinement
of this notion, called \emph{separated abstract HO-GSOS law}. 
Such a refinement is necessary because (unsurprisingly) an unrestricted small-step 
semantics need not generally correspond to a meaningful big-step semantics. Essentially this expressivity surcharge of small-step rule formats
has been previously acknowledged in the context of (failure of) congruence properties of
weak bisimilarity~\cite{Bloom95,Glabbeek11,TsampasWilliamsEtAl21,GoncharovMiliusEtAl22,UrbatTsampasEtAl23}. 
In fact, proving~\eqref{eq:sem-eq} first requires interpreting it meaningfully. The original notion of abstract HO-GSOS is too general for this purpose: it does not distinguish between values and non-values, nor does it generally support the definition of multi-step transitions~$\to^\star$. 
The separation requirement is introduced precisely to cater for this.
We then define an abstract 
notion of big-step semantics and establish an abstract formulation of equivalence~\eqref{eq:sem-eq}.
Yet, in order to prove~\eqref{eq:sem-eq}, restricting to separated abstract HO-GSOS is not enough,
which led us to a sufficient condition that we call \emph{strong separability}.
We elaborate on these matters by example in~\autoref{sec:back}. 

In our developments, we employ the versatile language of
\emph{category theory}, while our approach as a whole aligns closely with
\emph{functional semantics}~\cite{Danielsson12,Uustalu13},
as opposed to relational semantics. In functional semantics, we can view rules
such as those in \autoref{fig:cbn-lambda} as functional transformations of premises
to conclusions, and therefore, define the semantics of programs as certain fixpoints.
The equivalence to the familiar relational semantics is achieved by interpreting relations as nondeterministic functions
-- specifically, as effectful functions w.r.t. the powerset monad.

In summary, we contribute as follows.
\begin{itemize}
  \item We introduce the notion of separated abstract HO-GSOS law (\autoref{def:separated}) for modelling 
  small-step semantics, refining the existing abstract HO-GSOS laws~\cite{GoncharovMiliusEtAl23};
  \item We introduce and argue for the strong separation conditions (\autoref{def:strong-sep}),
  meant to guarantee that a given small-step semantics can have a big-step counterpart;
  \item We introduce an abstract notion of big-step semantics~(\autoref{def:big-step}); 
  \item We provide an abstract translation from small-step to big-step and establish \eqref{eq:sem-eq} under the strong separation assumption;
  \item We elaborate various instances of our abstract framework and the equivalence~\eqref{eq:sem-eq}~(\autoref{sec:exa}, \autoref{sec:lam}).
\end{itemize}
We implemented our notions and constructions in Haskell, as well as those examples 
from \autoref{sec:exa} that are hosted in the category $\Set$ of sets and functions.
Most of the proofs are placed in the appendix for space reasons. 

\myparagraph{Related Work}
The abstract (categorical) perspective on the small-step and big-step semantics we develop here
is enabled by recent advances in \emph{higher-order mathematical operational semantics}~\cite{GoncharovMiliusEtAl23},
establishing a connection between sets of operational semantics rules and certain (di\dash)natural transformations. 
The first-order form of this connection goes back to the seminal work of Turi and Plotkin~\cite{TuriPlotkin97}. 
Without this leverage, the question of the general connection between small-step and  
big-step semantics was addressed in the literature in a syntax-driven manner. Ciob{\^a}c{\u{a}}~\cite{Ciobaca13},
motivated similarly to us, proposed an automatic translation of small-step specifications to big-step specifications,
and essentially proved~\eqref{eq:sem-eq} using purely syntactic methods, under a number of assumptions different from ours. 
Similarly,~\citet{Bach-PoulsenMosses14} described a translation of small-step specifications to 
\emph{pretty-big-step specifications}, which are fundamentally a certain form of big-step specifications. 
Our categorical abstraction of operational semantics and their transformations are related to the ideas of functional 
semantics~\cite{Danielsson12,Uustalu13,OwensMyreenEtAl16}. A functional semantics essentially replaces (big-step) rules 
with functional transformations equipped with a clock to ensure totality. This clock can be integrated into a monad, such as the \emph{delay monad}~\cite{Capretta05}. 
Our treatment requires an $\omega$-continuous monad~\cite{GoncharovSchroderEtAl18}, as a parameter, instead, e.g.\ a partiality monad, 
which can be viewed as an extensional counterpart of the delay monad~\cite{AltenkirchDanielssonEtAl17}. 
In this paper, we are not focusing on constructive aspects, and -- mainly 
for the sake of simplicity -- stick to the powerset monad as the main example of $\omega$-continuous monad.

Special restricted rule formats for small-step semantics were proposed by~\citet{Bloom95} and~\citet{Glabbeek11},
to ensure congruence properties of behavioural equivalences. Similar restrictions later resurfaced
in previous abstract treatments of small-step semantics~\cite{TsampasWilliamsEtAl21,GoncharovMiliusEtAl22,UrbatTsampasEtAl23}. Our condition of \emph{strong separation}
abstracts similar restrictions -- it is analogous to claiming that the specification contains enough \emph{patience rules}~\cite{Bloom95,Glabbeek11}.

\begin{figure*}[t]
\begin{gather*}
\frac{}{S\xto{t}S'(t)}
\qquad\quad
\frac{}{S'(t)\xto{s}S''(t,s)}
\qquad\quad
\frac{}{S''(t,s)\xto{r}(tr)(sr)}
\\[2ex]
\frac{}{K\xto{t}K'(t)}
\qquad\quad
\frac{}{K'(t)\xto{s}t}
\qquad\quad
\frac{}{I\xto{t}t}
\qquad\quad
\frac{t\to t'}{t s\to t's}
\qquad\quad
\frac{t\xto{s} t'}{t s\to t'}
\end{gather*}
\caption{Small-step operational semantics of \xCL.}
\label{fig:rules}
\end{figure*}
\section{Abstract Higher-Order Operational Semantics: Overture}\label{sec:back}

As a simple motivating example we consider call-by-name \emph{extended combinatory 
logic} ($\xCL$) previously introduced for analogous purposes~\cite{GoncharovMiliusEtAl23}. 
This language is a variant of the well-known \textbf{SKI} calculus~\cite{Curry30}
and being computationally equivalent to the $\lambda$-calculus avoids the technical 
overhead of name management. The terms (i.e.\ programs) of \xCL are generated 
by the following grammar
\begin{align}\label{eq:xcl-grammar}
  t,s \Coloneqq S \mid K \mid I \mid
  S'(t) \mid K'(t) \mid S''(t,s) \mid t s.
\end{align}
Here, the binary application operator is as usual represented by juxtaposition, 
however, in the sequel, we will also use $\tcomp(t,s)$ as a verbose synonym to $ts$.
The letters $S$, $K$ and $I$ are the standard combinators, i.e.\ constants that represent
corresponding closed $\lambda$-terms. Their variants $S'$, $S''$ and~$K'$
capture partial applications of $S$ and $K$.

The small-step semantics rules for \xCL are displayed in~\autoref{fig:rules}. 
For example, we can derive the standard reduction for 
the $S$-combinator: $S\, t\, s\, r\to^\star (t\,r)(s\,r)$ (modulo addition of intermediate 
unlabeled transitions). The auxiliary labeled transitions $t\xto{s} t'$ represent the 
fact that~$t$ reduces to~$t'$ by consuming an argument $s$. Such use of labeled transitions
has previously been adopted by~\citet{Abramsky90} and Gordon~\cite{Gordon99}.

\subsection{Combinatory Logic in Haskell}\label{sec:hask}

\begin{figure}[t]
\begin{lstlisting}
data Free s x  = Res x | Cont (s (Free s x))     -- Free monad over given functor
type Initial s = Free s Void                     -- Initial algebra of given functor       

newtype Mrg s x = Mrg (s x x) 
sigOp = Cont . Mrg 

-- Type class for HO-GSOS, parametrized by signature and behaviour
class (Bifunctor s, MixFunctor b) => HOGSOS s b where 
  -- Abstract representation of small-step rules
  rho :: s (x, b x y) x -> b x (Free (Mrg s) (Either x y))

  -- Operational model: derivable abstract semantics function
  gamma :: Initial (Mrg s) -> b (Initial (Mrg s)) (Initial (Mrg s))
  gamma (Cont (Mrg t)) = mx_second (>>= nabla) $ rho $ first (id &&& gamma) t
    where nabla = either id id

-- Instantiations for xCL:

data XCL' x y = S | K | I | S' x | K' x | S'' x x | Comp x y    -- Signature
type XCL      = Mrg XCL'                                    

data Beh x y = Eval (x -> y) | Red y                            -- Behaviour

instance HOGSOS XCL' Beh where
  rho :: XCL' (x, Beh x y) x -> Beh x (Free XCL (Either x y))
  rho S = Eval $ sigOp . S' . Res . Left
  rho K = Eval $ sigOp . K' . Res . Left
  rho I = Eval $ Res . Left

  rho (S' (s, _)) = Eval $ \t -> sigOp $ S'' (Res $ Left s) (Res $ Left t)
  rho (K' (s, _)) = Eval $ \t -> Res $ Left s

  rho (S'' (s, _) (u, _)) = 
    Eval $ \t -> sigOp $ Comp (sigOp $ Comp (Res $ Left s) (Res $ Left t)) 
                              (sigOp $ Comp (Res $ Left u) (Res $ Left t))

  rho (Comp (_, Red s) u)  = Red $ sigOp $ Comp (Res $ Right s) (Res $ Left u)
  rho (Comp (_, Eval f) u) = Red $ Res (Right $ f u)
\end{lstlisting}
\caption{\lstinline{HOGSOS} type class and \lstinline{XCL'} as its instance.}
\label{fig:code1}
\end{figure}

To build intuition for the upcoming technical developments, we present a semi-formal exposition of \xCL, including its small-step and big-step operational semantics, as well as their relationship, using Haskell. This exposition serves to motivate and clarify the original notion of higher-order abstract GSOS laws~\cite{GoncharovMiliusEtAl23} and their separated variant. A rigorous categorical treatment is provided in \autoref{sec:cat}, to which readers primarily interested in the formal theory may skip directly.

Consider the 
(incomplete) Haskell code in \autoref{fig:code1}. Assuming that \lstinline{s} models a signature, \lstinline{Free s}
and \lstinline{Initial s} model terms over this signature with and without variables 
respectively.
The central definition is that of \lstinline{HOGSOS}, which 
is a type class, parameterized by a signature bifunctor~\lstinline{s}
and a mixed variant behaviour functor \lstinline{b} (more precisely, \lstinline{b} 
is contra-variant in the first argument and co-variant in the second). We explain 
later why the signature functor has two arguments (i.e.\ is a bifunctor) -- the 
standard notion of signature is recovered as \lstinline{Mrg s}.

The \lstinline{HOGSOS} class has one method \lstinline{rho},
which encodes the rules of operational semantics. How exactly \lstinline{rho}
does that is illustrated by instantiating it to the case of \xCL (the code stating 
that \lstinline{XCL'} and \lstinline{Beh} are a bifunctor and a mixed variance
functor correspondingly is omitted). The argument of 
\lstinline{rho} refers to a relevant signature symbol (via \lstinline{s}),
to its arguments that match the left bottom part of the corresponding rule (via \lstinline{x}),
and to the corresponding argument's behaviours that match the right top parts of the 
corresponding rule (via \lstinline{b x y}). The \emph{operational model} \lstinline{gamma}
for every choice of \lstinline{s}, \lstinline{b} and \lstinline{rho} produces 
the semantics of a given closed term, which in the case of \xCL corresponds to the 
transitions $p\to q$ or $p\xto{t} q$ derivable with the rules in~\autoref{fig:rules}. 
The recursive definition of \lstinline{gamma} uses the fact that \lstinline{Free s}
is a monad and calls the operation~\lstinline{&&&} for pairing functions and  
functorial actions \lstinline{first}, \lstinline{mx_second} of \lstinline{s}
and $b$ on the first and the second argument correspondingly. Recall that \lstinline{$}
reads as function composition and the inlined definition of \lstinline{nabla} 
captures the universal map from the coproduct of \lstinline{x} with itself to
\lstinline{x} yielding \lstinline{>>= nabla :: Free (Mrg s) (Either (Free (Mrg s)) (Free (Mrg s))) -> Free (Mrg s)}.

\begin{figure*}[t]
\begin{lstlisting}
-- Signature functor as a sum of value and computation parts
data SepSig' sv sc x y = SigV (sv y) | SigC (sc x y) 
type SepSig sv sc      = Mrg (SepSig' sv sc)

type InitialV sv sc = sv (Initial (SepSig sv sc))
type InitialC sv sc = sc (Initial (SepSig sv sc)) (Initial (SepSig sv sc))

-- Form of the behaviour functor in the separable setting
data SepBeh d x y = BehV (d x y) | BehC y

-- Type class for separated HO-GSOS
class (MixFunctor d, Functor sv, Bifunctor sc) => SepHOGSOS sv sc d where
  rhoV :: sv x -> d x (Free (SepSig sv sc) x) 
  rhoC :: sc (x, SepBeh d x y) x -> Free (SepSig sv sc) (Either x y)

  rhoCV :: sc (x, d x y) x -> Free (SepSig sv sc) (Either x y)
  rhoCV = rhoC . first (second BehV)

  gammaC :: Proxy d -> InitialC sv sc -> Initial (SepSig sv sc)
  gammaC (p :: Proxy d) t = (rhoC @_ @_ @d $ first (id &&& gamma) t) >>= nabla
    where nabla = either id id

  beta :: (Functor sv, Bifunctor sc, MixFunctor d, SepHOGSOS sv sc d) 
    => Proxy d -> InitialC sv sc -> InitialV sv sc
  beta (p :: Proxy d) t = case gammaC p t of Cont (Mrg (SigV t)) -> t ; 
                                             Cont (Mrg (SigC t)) -> beta p t
  
instance (SepHOGSOS sv sc d) => HOGSOS (SepSig' sv sc) (SepBeh d) where
  rho :: SepSig' sv sc (x, SepBeh d x y) x -> SepBeh d x (Free (SepSig sv sc) (Either x y))
  rho (SigV v) = BehV $ (right $ fmap Left) $ rhoV v
  rho (SigC c) = BehC $ rhoC c
\end{lstlisting}
\caption{\lstinline{SepHOGSOS} type class as a refinement of \lstinline{HOGSOS}.}
\label{fig:code2}
\end{figure*}

The behaviour functor \lstinline{Beh} is a coproduct of two functors: the first one
caters for labeled transitions (evaluations), the second one caters for unlabeled 
transitions (reductions). The following \lstinline{instance} declaration captures 
the rules from~\autoref{fig:rules}. The separation on evaluations and reductions is not present on the abstract level of
\lstinline{HOGSOS} and \lstinline{rho}, and hence they are not suitable for defining 
\emph{multi-step semantics} $\to^\star$ involved in~\eqref{eq:sem-eq}. We thus 
introduce a refinement of abstract HO-GSOS in \autoref{fig:code2} and call it 
\emph{separated abstract HO-GSOS}. 
That is, we postulate a partitioning of both the signature functor \lstinline{SepSig} 
and the behaviour functor \lstinline{SepBeh} into a \emph{``value part``} and a 
\emph{``computation part``}, which is indicated by appending \textit{V} and \text{C} 
correspondingly. The computation part of the signature functor for \xCL is 
the application operator, and the computation part of the behaviour functor for \xCL 
is the part of unlabeled transitions.

It follows from the above code that 
\lstinline{SepHOGSOS} instantiates \lstinline{HOGSOS}, which is shown by assembling a suitable abstract 
HO-GSOS rule \lstinline{rho} from \lstinline{rhoV} and \lstinline{rhoC}. For technical 
reasons (to ensure unambiguous type checking) we use explicit type applications via~\lstinline{@} -- those can be safely ignored in reading.
We thus inherit the operational model \lstinline{gamma}, which, in turn, yields its own 
computation part \lstinline{gammaC}, thanks to the separability assumption. Using 
\lstinline{gammaC}, we recursively define the abstract multi-step semantics~\lstinline{beta}. 
Here, the sake of unambiguous type checking, we use Haskell's mechanism of proxies,
i.e., in this case, the dummy argument \lstinline{p}, which is only needed to facilitate  type checking.

\begin{figure*}[t]
\begin{lstlisting}
-- Type class for big-step SOS
class (Functor sv, Bifunctor sc) => BSSOS d sv sc where
  xi :: sc (sv x) x -> Free (SepSig sv sc) x

  zetahat :: Initial (SepSig sv sc) -> InitialV sv sc
  zetahat (Cont (Mrg (SigV v))) = v
  zetahat (Cont (Mrg (SigC c))) = zetahat @d $ join $ xi @d $ first (zetahat @d) c

  zeta :: InitialC sv sc -> InitialV sv sc
  zeta = zetahat @d . sigOp . SigC

instance (SepHOGSOS sv sc d) => BSSOS d sv sc where
  xi :: sc (sv x) x -> Free (SepSig sv sc) x
  xi t = rhoCV (bimap ((sigOp . SigV &&& mx_second @d join . rhoV) . fmap return) return t) 
	      >>= nabla
  where nabla = either id id
\end{lstlisting}
\caption{\lstinline{BSSOS} type class, and \lstinline{SepHOGSOS} as its instance.}
	\label{fig:code3}
\end{figure*}

Finally, in order to interpret the left-hand side of the equivalence~\eqref{eq:sem-eq}
and the equivalence itself, we need to define big-step semantics abstractly
and to link it to the small-step semantics. We do this as shown in \autoref{fig:code3}.
The notion of big-step semantics is formalized with the class \lstinline{BSSOS} and 
its method \lstinline{xi}.
Again, for technical reasons we let \lstinline{BSSOS} vacuously depend on~\lstinline{d}
-- this is needed for the following \lstinline{instance} declaration to type check.
The derivable map \lstinline{zetahat} and its variant \lstinline{zeta} abstractly define the 
evaluation relation $\Dar$. From an HO-GSOS specification we automatically obtain 
a big-step SOS specification in a slightly wordy, but essentially simple manner. 

In the case of \xCL we obtain the specification displayed in~\autoref{fig:big-ski}.
This specification can facilitate understanding the type of \lstinline{xi}: unless the
rule is an axiom of the form $v\Dar v$ where $v$ is a term whose topmost 
operator is from \lstinline{cv}, the conclusion of the rule has the form $pq\Dar v$ and
the premise of the rule contains the judgement of the form $t\Downarrow v$ with $t$ 
determined by $q$ and by such~$w$ that $p\Dar w$ also occurs in the premise. The 
latter type of rules is thus determined by a choice of an operation from \lstinline{sc},
by a choice of an operation from \lstinline{sv} for every of its argument that corresponds to 
the first position of \lstinline{sc} as a functor, and by the term $t$ that can parametrically
depend on the arguments of these operations. This is, in a nutshell, the information that 
\lstinline{xi} carries. Note that the rules in~\autoref{fig:big-ski} marked by 
an asterisk can be simplified in the obvious way by removing premises of the form 
$w\Dar v$ and by replacing $v$ with~$w$ in the conclusions. 
 
The rules in~\autoref{fig:big-ski} also help one to see why the signature 
functor \lstinline{sc} has two arguments. This concretely means partitioning
the arguments of any operation in \lstinline{sc} over two types: \emph{strict} 
and \emph{lazy}. In \xCL \lstinline{sc} contains precisely one operator -- application, 
whose first argument is strict, and whose second argument is lazy. As the rules 
in~\autoref{fig:big-ski} illustrate, we allow evaluation of strict arguments only. Unlike 
small-step semantics, we cannot simply add judgements of the form $q\Dar w$ even if 
we do not use $w$. For example, given any diverging term $\Omega$, we must have 
$K I\Omega\Dar I$, which would not be the case if we conditioned this on the existence of a derivation~$\Omega\Dar w$. 
 
We can now express~\eqref{eq:sem-eq} as the equality
\begin{equation}\label{eq:equiv-abst}
\text{\lstinline{beta p t == zeta t}}
\end{equation}
for all \lstinline{p :: Proxy d} and \lstinline{t :: InitialC sv sc}.

For a final note, observe that~\eqref{eq:equiv-abst}, albeit desirable, need not 
always be true.
\begin{figure*}[t]
	\begin{flalign*}
		\makebox[0pt][l]{\textit{Values:}}&&  v,w&\Coloneqq I\mid K\mid S\mid K'(t)\mid S'(t)\mid S''(s,t)&\hspace{3cm}\\ 
		\makebox[0pt][l]{\textit{Terms:}}&&   s,t,r,q&\Coloneqq v\mid s t&\\[-3ex]
	\end{flalign*}
	\begin{gather*}
		\infer{v \Dar v}{} \qquad\qquad
		\infer{st\Dar v}{s \Dar I & t \Dar v} \qquad\qquad
		\infer{st\Dar v}{s \Dar K & K'(t)\Dar v}~{\rule{0pt}{12pt}{}}^* \qquad\qquad
		\infer{st\Dar v}{s \Dar S& S'(t)\Dar v}~{\rule{0pt}{12pt}{}}^* \\[2ex] 
		\infer{st\Dar v}{s \Dar K'(r) & r \Dar v} \qquad\qquad
		\infer{st\Dar v}{s \Dar S'(r)&S''(r,t)\Dar v}~{\rule{0pt}{12pt}{}}^*\qquad\qquad 
		\infer{st\Dar v}{s \Dar S''(r,q) & (rt)(qt)\Dar v}
	\end{gather*}		
	\caption{Big-step operational semantics of \xCL.}
	\label{fig:big-ski}
\end{figure*}
\begin{example}\label{exa:non-sep}
Consider a language over the signature $\{f/1,g/1,\Omega/0\}$ where $/n$ indicates 
that the arity of the corresponding operation is $n$. Consider the following 
small-step specification:
\begin{gather*}		
		\frac{}{g(x) \stackrel{y}{\to} f(y)} \qquad\qquad
		\frac{}{\Omega \to \Omega}  \qquad\qquad
		\frac{x \to y}{f(x) \to g(y)} \qquad\qquad
		\frac{x \stackrel{x}{\to} y}{f(x) \to x} 
\end{gather*}		
which identifies $g$ as a value former and $f$ and $\Omega$ as computation formers.
This specification yields a separated abstract HO-GSOS law, in which the only 
argument of $f$ is necessarily strict, because the behaviour of $f(t)$ generally depends 
on the behaviour of $t$. The only way to define big-step rules would be as follows:
\begin{gather*}
\infer{g(x) \Dar g(x)}{} \qquad\qquad
\infer{f(x) \Dar v}{x \Dar g(y)&\qquad g(y)\Dar v}
\end{gather*}
The equivalence~\eqref{eq:sem-eq} now fails, because $f(f(g(\Omega))) \to g(g(\Omega))$, 
but $f(f(g(\Omega))) \Dar g(\Omega)$. 
\exend\end{example} 

\subsection{Categorical Modeling} \label{sec:cat}
The reasoning of the previous section is not sufficiently precise in various respects. 
Haskell provides a very concrete type-theoretic ambient with general recursion
and other features we might want or not want to include, but in any case we need 
to be conscious about them. This is particularly important for constructing proofs,
an aspect, we completely omitted so far, but which we consider as the main contribution.
The above treatment of \xCL can be naturally formalized in the category of sets,
using a suitable \emph{partiality monad}~\cite{ChapmanUustaluEtAl15,AltenkirchDanielssonEtAl17,Goncharov21} for modeling iteration and recursion. 
More generally, one would need other categories: multisorted sets for typed languages,
categories of presheaves or nominal sets for the $\lambda$-calculus, the corresponding 
combinations and extensions thereof. It is also not necessary to restrict to partiality
as the only effect -- one can treat nondeterministic or even probabilistic semantics 
in a similar manner. We thus generally work with \emph{strong $\omega$-continuous monads},
which are arguably the largest semantically relevant class of monads that support iteration 
via least fixed points~\cite{GoncharovSchroderEtAl18}.

We then recall and/or fix relevant categorical notations and conventions in order. 
For the basics of category theory we refer to~\cite{Mac-Lane78,Awodey10}. 
In a category~$\BC$,~$|\BC|$ will denote the 
class of objects and~$\BC(X,Y)$ will denote the set of morphisms from~$X$ to~$Y$. The judgement 
$f\c X\to Y$ will be regarded as an equivalent to $f\in\BC(X,Y)$ if $\BC$ is clear
from the context.
We will denote by $\id_X$, or simply $\id$ the identity morphism on~$X$. 
In what follows, we generally work in an ambient distributive category~$\BC$ (as defined below). 
By $|\BC|$ we refer to the objects of $\BC$. We tend to
omit indexes at natural transformations for better readability. 
\emph{Dinatural transformations} generalize natural transformations to 
the case of mixed variance functors. More precisely, given two functors $F,G\c\BC^\op\times\BC\to\BD$,
a family of morphisms $\alpha = (\alpha_{X,Y}\c F(X,Y)\to G(X,Y))_{X,Y\in |\BC|}$ is a dinatural transformation
if the diagram
\begin{equation*}
\begin{tikzcd}[column sep=3em, row sep=1em]
&
F(X,X)
	\rar["{\alpha_{X,X}}"]
&
G(X,X)
	\ar[dr,"{G(\id,f)}"]
\\
F(Y,X) 
	\ar[ur,"{F(f,\id)}"]
    \ar[dr,"{F(\id,f)}"'] 
	& & & 
G(X,Y)
\\
&
F(Y,Y)
	\rar["{\alpha_{Y,Y}}"]
&
G(Y,Y)
	\ar[ur,"{G(f,\id)}"']
\end{tikzcd}
\end{equation*} 
commutes for any $f\c X\to Y$. A prototypical example of a dinatural transformation is the evaluation transformation $Y^X\times X\to Y$.

Some further notions we will need are as follows.

\myparagraph{Distributive categories} %
A \emph{distributive category} is a category with finite products and coproducts, 
and such that every morphism $[\id\times\inl,\id\times\inr]\c X\times Y + X\times Z\to X\times (Y+Z)$ is an isomorphism.
Let $\nabla=[\id,\id]\c X+X\to X$ and let $\dl$ be the functor $\BC\to\BC\times\BC$,
sending~$X$ to $(X,X)$. Moreover, let $\Fst,\Snd\c\BC\times\BC\to\BC$ be the 
obvious projections.

\myparagraph{Functors and algebras} We involve three types of (endo-)functors: the usual 
unary covariant functors $F\c\BC\to\BC$, bifunctors $F\c\BC\times\BC\to\BC$ and
mixed variance functors $F\c\BC^\op\times\BC\to\BC$. Given a functor $F\c\BC\to\BC$
in a distributive category~$\BC$, the \emph{(pointwise) free monad} $F^\star\c\BC\to\BC$ 
over $F$ is characterized by a universal property: for any object $X$ there are 
morphisms $\iota_X\c F(F^\star X)\to F^\star X$, and $\eta_X\c X\to F^\star X$, such 
that for any $f\c FY\to Y$ and $g\c X\to Y$ the diagram 
\begin{equation*}
\begin{tikzcd}[column sep=3em, row sep=normal]
F(F^\star X)
	\rar["\iota_X"] 
	\dar["{F(\init[f,g])}"']
&
F^\star X
	\dar["{\init[f,g]}"]
&[4em]
X
\lar["\eta_X"']
\ar[dl,"g"]
\\
FY
	\rar["f"] 
&
Y
&
\end{tikzcd}
\end{equation*}
commutes for precisely one morphism $\init[f,g]$. It follows that $\iota_X$ and $\eta_X$
are natural in~$X$. By generalities (Lambek's theorem), $[\eta,\iota]$ is an isomorphism. 
For every $X$, $F^\star X$ is also called a \emph{free (F-)algebra}
on $X$. We can often think of~$F^\star X$ as an object of terms with variables from $X$
and with operations from~$F$, and of $\mu F=F^\star\iobj$ as an object of closed terms thereof. 
This is specifically true for \emph{polynomial functors},
i.e.\ functors of the form $FX=\coprod_{f\in Ops} X^{\ar(f)}$ where $Ops$ is a 
set of operations~$f$ of corresponding arities $\ar(f)$. One example is the functor 
\begin{align*}
  F X
  = \underbrace{1+1+1}_{I,K,S} \,+\, \underbrace{X + X}_{K',S'}
  \,+\, 
  \underbrace{X\times X}_{S''}%
\end{align*}
induced by the grammar of \xCL~\eqref{eq:xcl-grammar}. 
By definition, an element of the dependent sum $\coprod_{f\in Ops} X^{\ar(f)}$ 
has the form $(f\in Ops,(x_i\in X)_{i\in\ar(f)})$, which we will also write 
as $f(x_i)_{i\in\ar(f)}$, or even $f(x_1,\ldots,x_{\ar(f)})$ when $\ar(f)\in\nats$,
if no confusion arises. In what follows we assume that all the involved free objects 
$F^\star X$ exist, without further mention.

\myparagraph{Strong monads and distributive laws} A monad~$\BBT$ on~$\BC$ is determined by a \emph{Kleisli triple}~$(T,\eta, (-)^\klstar)$,
consisting of a map $T\c{|\BC|\to|\BC|}$, a family of morphisms 
$(\eta_X\c X \to TX)_{X\in|\BC|}$ and \emph{Kleisli lifting} sending each~$f\c X \to T Y$ to~$f^\klstar\c TX \to TY$ and
obeying \emph{monad laws}: 
\begin{align*}
	\eta^{\klstar}=\id,\qquad f^{\klstar}\comp\eta=f,\qquad (f^{\klstar}\comp g)^{\klstar}=f^{\klstar}\comp g^{\klstar}.
\end{align*}
It follows that~$T$ extends to a functor, $\eta$ extends to a natural transformation -- \emph{unit},
${\mu = \mmul\c TTX\to TX}$ extends 
to a natural transformation -- \emph{multiplication}, and that $(T,\eta,\mu)$ 
is a monad in the standard sense~\cite{Mac-Lane78}. Any free monad $F^\star$ is
an example of a monad. 

We will emphasize that a functor $T$ is also a monad by writing 
it boldfaced (such as~$\BBT$).
Given $f\c X\to TZ$ and $g\c Y\to TW$, we abbreviate $f\klplus g = [T\inl\comp f,T\inr\comp g]\c X+Y\to T(Z+W)$.

A monad $\BBT$ is \emph{strong} if it comes with 
a natural transformation $\tau_{X,Y}\c X\times TY\to T(X\times Y)$ called \emph{strength} 
and satisfying a number of coherence conditions~\cite{Moggi91}. 
A well-known fact due to Kock~\citep{Kock72} is that in self-enriched categories (such 
as $\Set$) strength is equivalent to enrichment. In particular, in $\Set$, 
every functor $T$ and every monad $\BBT$ are strong with the canonical strength 
$\tau_{X,Y} = \lambda (x,z).\,T(\lambda y.\,(x,y))(z)$.

By a \emph{(Kleisli) distributive law} between
a monad $\BBT$ and a functor $F$ we mean a natural transformation $FT\to TF$ suitably interacting 
with unit and multiplication of the monad~\cite{HasuoJacobsEtAl07}.

\myparagraph{Order enrichment and fixpoints} 
Recall that \emph{Kleene's fixpoint theorem} states that every continuous endomap 
$f$ on a pointed $\omega$-cpo has the least pre-fixpoint~$\mu f$ (which is also the least fixpoint),
which is a least upper bound of the chain
\begin{align*}
\bot\appr f(\bot)\appr f(f(\bot))\appr\ldots
\end{align*}
An \emph{$\omega$-continuous} monad~\cite{GoncharovSchroderEtAl18} is a monad~$\BBT$ together with
an enrichment of the Kleisli category~$\BC_{\BBT}$ of~$\BBT$ over pointed $\omega$-cpos and
(nonstrict) $\omega$-continuous maps, satisfying the following principles:
\begin{itemize}
\item strength is $\omega$-continuous: $\tau(\id\times\bigjoin_i f_i)=\bigjoin_i\tau(\id\times f_i)$;
\item copairing in $\BC_{\BBT}$ is $\omega$-continuous: $\bigl[\bigjoin_i f_i,\bigjoin_i g_i\bigr]= \bigjoin_i [f_i,g_i]$;
\item bottom elements are preserved by strength and by postcomposition in $\BC_{\BBT}$: $\tau(\id\times\bot)=\bot$, $f^\klstar\comp\bot=\bot$.
\end{itemize}
Every $\omega$-continuous monad $\BBT$ is a (complete) Elgot monad, i.e.\ it supports
an (Elgot) iteration operator that sends every $f\c X\to T(Y+X)$ to $f^\istar\c X\to TY$,
subject to several standard laws of iteration. Specifically, $f^\istar=\mu g.\, [\eta,g]^\klstar\comp f$. 
Standard classical examples of (strong) $\omega$-continuous monads are the \emph{maybe-monad} $TX=X+1$
and the \emph{powerset monad} $TX=\PSet X$. 
Note that the identity monad $T=\Id$ is generally not $\omega$-continuous, for
the Kleisli hom-sets need not posses least elements.
We call a distributive law $\dl\c F\BBT\to\BBT F$ $\omega$-continuous if all the correspondences $f\in\BC(X,TY)\mapsto\dl_Y\comp Ff\in\BC(FX, TFY)$
are $\omega$-continuous.

\subsection{Abstract HO-GSOS}
We proceed to recall and modify slightly the notion of abstract 
HO-GSOS law from previous work~\cite{GoncharovMiliusEtAl23}. This is indeed the notion
we already implemented in~\autoref{sec:hask}. Let us fix a 
bifunctor $\Sigma'\c\BC\times\BC\to\BC$ (signature) and a mixed variance 
functor $B\c\BC^\op\times\BC\to\BC$ (behaviour) on a distributive category~$\BC$. 
Let $\Sigma=\Sigma'\dl$ and let
\begin{align}\label{eq:ho-gsos-law}
\rho_{X,Y} \c \Sigma' (X \times B(X,Y),X)\to B(X, \Sigma^\star (X+Y)),
\end{align}
be a family of morphisms natural in $Y$ and dinatural in $X$. 
This is a slight generalization of the original notion, which is obtained
making~$\Sigma'$ independent of the second argument. In that case $\Sigma'(X,Y)=\Sigma X$.
We need the additional parameter for $\Sigma'$ to identify those arguments of signature 
operations, whose behaviour is not inspected. The key notion 
that~\eqref{eq:ho-gsos-law} generates is that of \emph{operational model}, which is a morphism 
$\gamma$, defined by parametrized structural recursion as follows:
\begin{equation}\label{eq:op-mod}
\kern-2ex\begin{tikzcd}[column sep=3ex, row sep=4ex]
\Sigma'(\mS,\mS)
  \dar["{\Sigma'(\brks{\id,\gamma},\id)}"']
  \ar[rr,"\iota"] & &[5ex]
\mS
  \dar["\gamma"]\\
\Sigma'(\mS\times {B(\mS,\mS)},\mS)
  \rar["\rho"] &
B(\mS,\Sigma^\star (\mS+\mS))
  \rar["{B(\id,\nabla^\klstar)}"] &
B(\mS,\mS)
\end{tikzcd}
\end{equation}
That is: there is precisely one $\gamma\c\mS\to B(\mS,\mS)$, such that~\eqref{eq:op-mod} commutes.

In fact, given a law~\eqref{eq:ho-gsos-law}, we can introduce an abstract higher-order 
GSOS for $\Sigma$ and $B$ in the original sense as follows:
\begin{equation}\label{eq:ho-var-conv}
\begin{tikzcd}[column sep=20ex,row sep=4ex]
\Sigma' (X \times B(X,Y),X \times B(X,Y)) 
	\rar["{\Sigma' (\id\times B(\id,\inr),\inl\comp\fst)}"]
&
\Sigma' (X \times B(X,X+Y),X + Y)
	\ar[dl,"\rho"',bend right=4,pos=.7]\\
B(X, \Sigma^\star (X+(X+Y)))
	\rar["{B(\id,\Sigma^\star[\inl,\id])}"]
&
B(X, \Sigma^\star (X+Y))
\end{tikzcd}
\end{equation}
and it is shown below that the operational model for it coincides with the one specified 
by~\eqref{eq:op-mod}. In that sense our present generalization is indeed mild -- it does 
not add more strength to the original notion, but it provides more flexibility for 
the analysis of the transformations~\eqref{eq:ho-gsos-law}.

\begin{proposition}\label{pro:theta-om}
Let $\rho'_{X,Y}\c \Sigma(X \times B(X,Y)) \to B(X, \Sigma^\star (X+Y))$ be defined 
by~\eqref{eq:ho-var-conv}. Then the operational model for $\rho'$ is the unique such 
morphism that the diagram~\eqref{eq:op-mod} commutes.
\end{proposition}

\section{Separable Abstract Higher-Order GSOS}\label{sec:gsos-separated}

Let us reintroduce the notion of 
separable abstract HO-GSOS from~\autoref{sec:hask} in categorical terms.
\begin{definition}[Separable Abstract Higher-Order GSOS]\label{def:separated}
We say that the law~\eqref{eq:ho-gsos-law} is \emph{separable} if $B(X,Y)=D(X,Y) + TY$, $\Sigma'=\Sigma_\val\Snd+\Sigma_\com$
and 
\begin{align}\label{eq:separated-rho}
\rho = D(\id,\Sigmas\inl)\comp \rho^\val+\rho^\com %
\end{align}
for some $D\c\BC^\op\times\BC\to\BC$, $\Sigma_\val\c\BC\to\BC$, $\Sigma_\com\c\BC\times\BC\to\BC$,
a strong monad $\BBT$, families of morphisms
\begin{align}
\label{eq:separated1}  
\rho_{X}^\val\c&\Sigma_\val X \to D(X,\Sigma^\star X),
\\*
\label{eq:separated2}
\rho_{X,Y}^\com\c&\Sigma_\com(X\times B(X,Y),X)\to T\Sigma^\star (X+Y),
\intertext{dinatural in $X$ and natural in $Y$, and a distributive law}
\label{eq:separated3}
\dl_{X,Y}\c&\Sigma_\com(TX,Y)\to T\Sigma_\com(X,Y).
\end{align}
between the monad $\BBT$ and the functors $\Sigma_\com(\argument,Y)$ naturally in $Y$.
The triple $(\rho^\val,\rho^\com,\dl)$ is then a \emph{separated abstract higher-order GSOS law}.

We call $\Sigma_\val$ \emph{value formers}
and $\Sigma_\com$ \emph{computation formers}. Moreover, we call  
$\mSv=\Sigma_\val(\mS)$ the \emph{object of values} and $\mSc=\Sigma_\com(\mS,\mS)$
the \emph{object of computations}. Let $\iota^\val = \iota\comp\inl\c\Sigma_\val\Sigma^\star\to\Sigma^\star$
and $\iota^\com = \iota\comp\inr\c\Sigma_\com\Delta\Sigma^\star\to\Sigma^\star$.
\end{definition}
For the time being we are not making any assumptions about the monad $\BBT$. 
In the simplest (total, deterministic) case, $\BBT$ in \autoref{def:separated} is the identity 
monad and $\dl=\id$. %

Since $\mS\iso\mSv+\mSc$, the object of 
all closed $\Sigma$-terms $\mS$ crisply decomposes into values and computations.
That $\Sigma_\com$ is a binary functor is meant to capture a partitioning of the 
arguments of the computation formers into those that depend on the associated 
behaviour, and those that do not. We call the former type of arguments \emph{strict}
and the latter \emph{lazy}.

For a separated abstract higher-order GSOS law $(\rho^\val,\rho^\com,\dl)$,
we define a refinement $(\gamma^\val,\gamma^\com)$ of the operational model~\eqref{eq:op-mod}.
The morphism $\gamma^\val$ is the composition
\begin{align}\label{eq:op-mod-v}
\mSv
	\xto{\rho^\val}
D(\mS,\Sigma^\star\mS)
	\xto{D(\id,\mmul)} 
D(\mS,\mS)
\end{align}
and the morphism $\gamma^\com$ is characterized by the diagram
\begin{equation}\label{eq:op-mod-c}
\begin{tikzcd}[column sep=5ex, row sep=4ex]
\Sigma_\com(\mSv+\mSc,\mS)
	\dar["{\Sigma_\com(\brks{\iota,\gamma^\val+\gamma^\com},\id)}"']
	\ar[rr,"{\Sigma_\com(\iota,\id)}"] 
& 
&[2ex]
\mSc
	\dar["\gamma^\com"]	
\\
\Sigma_\com(\mS\times {B(\mS,\mS)},\mS)
  \rar["{\rho^\com }"] 
&
T\Sigmas(\mS+\mS)
	\rar["T\nabla^\klstar"]
&
T\mS
\end{tikzcd}
\end{equation}
in the following sense.
\begin{proposition}\label{pro:gamma_c}
There is unique $\gamma^\com\c\mSc\to T(\mS)$, for which \eqref{eq:op-mod-c}
commutes. Moreover:
\begin{align}\label{eq:gammas}
\gamma^\val+\gamma^\com = \gamma\comp\iota.
\end{align}
\end{proposition}
Let $\hat\gamma^\com$ be the morphism $[\eta\comp\iota^\val,\gamma^\com]\comp\iota^\mone\c\mS\to T(\mS)$.
Intuitively, $\hat\gamma^\com$ acts on computations as $\gamma^\com$ and as $\eta$ on values.

\begin{example}[Extended Combinatory Logic]\label{exa:xcl}
Recall the grammar~\eqref{eq:xcl-grammar} of the extended combinatory logic \xCL. 
The corresponding signature functor consists of two parts: 
\begin{displaymath}
  \Sigma_\val X = \coprod\nolimits_{f\in\{S,K,I,K',S',S''\}} X^{\ar(f)} \qquad\qquad \Sigma_\com (X,Y) = \{\tcomp\}\times (X\times Y).
\end{displaymath}
Here $\ar(f)$ denotes the arity of $f$. Binary application operator $\tcomp$
is the only computation former. The expression for $\Sigma_\com$ indicates that 
the one-step behaviour of it only depends on the behaviour of the first argument, but not on
the second. 

Let $D(X,Y) = Y^X$ and $T=\Id$. The small-step operational semantics rules in~\autoref{fig:rules} define $\rho^\val$
and $\rho^\com$, and hence the law~\eqref{eq:ho-gsos-law}. Concretely (eliding the obvious isomorphisms and parentheses):
\begin{align*}
\rho^\val(I)(r) 		=&\; r\qquad	&\rho^\val(K'(t))(r)  =&\; t	& \rho^\com(\tcomp((t,t'), s))	=&\; \tcomp(t', s)\\ 
\rho^\val(K)(r) 		=&\; K'(r)		&\rho^\val(S'(t))(r)	=&\; S''(t, r)	& \rho^\com(\tcomp((t,f), s)) =&\; f(s) \\
\rho^\val(S)(r) 		=&\; S'(r)		&\rho^\val(S''(t,s))(r)	=&\; \makebox[4em][l]{$\tcomp(\tcomp(t, r),\tcomp(s,r))$}\\[-8ex]
\end{align*}
\exend
\end{example}
From~\eqref{eq:separated2} we derive
\begin{align}
\label{eq:separated21}
\rho_{X,Y}^{\com\val}\c&\Sigma_\com(X\times D(X,Y),X)\xto{\rho^\com\comp\Sigma_\com(\id\times\inl,\id)}  T\Sigma^\star (X+Y). %
\end{align}
In what follows, we globally make the following mild technical assumption.
\begin{assumption}\label{ass:compl}
We assume that for all $X,Y,Z$, and for $\Sigma_\com$ figuring in~\autoref{def:separated} 
the morphisms $\Sigma_\com(\inl,\id)\c\Sigma_\com(X,Z)\to\Sigma_\com(X+Y,Z)$ are \emph{complemented},
and the complementation is natural in $X,Y$ and $Z$.
In other words, for some functor~$\Theta\c\BC\times\BC\times\BC\to\BC$ and some natural transformation 
$\theta_{X,Y,Z}\c\Theta(X,Y,Z)\to\Sigma_\com(X+Y,Z)$, the cospan 
\begin{align*}
\Sigma_\com(X,Z) \xto{\Sigma_\com(\inl,\id)}\Sigma_\com(X+Y,Z)\xfrom{\theta}\Theta(X,Y,Z)
\end{align*}
is a coproduct naturally in $X$, $Y$ and $Z$.
\end{assumption}
\begin{remark}\label{rem:poly-theta}
Note that if $\Sigma_\com$ is polynomial $\Sigma_\com(X,Y)=\coprod_{f\in Ops} X^{\sar(f)}\times Y^{\lar(f)}$ (where $\sar(f)$
is the number of strict arguments of $f$ and $\lar(f)$ is the number of lazy ones) 
\autoref{ass:compl} is satisfied whenever it is satisfied for every summand $X^{\sar(f)}\times Y^{\lar(f)}$.
Let us fix~$f$, and let $n=\sar(f)$, $m=\lar(f)$. \autoref{ass:compl} is then satisfied, since (by using the binomial formula)
\begin{align*}
(X+Y)^n\times Z^m\cong X^n\times Z^m + \coprod_{k=1}^n C_n^k\times Y^k\times X^{n-k}\times Z^m,
\end{align*}
and we can take $\Theta(X,Y,Z) = \coprod_{k=1}^n C_n^k\times Y^k\times X^{n-k}\times Z^m$.
\end{remark}

We now introduce a well-behavedness condition on separated abstract HO-GSOS 
guaranteeing that a notion of big-step operational semantics can be sensibly 
derived.

Recall that we defined $f\klplus g = [T\inl\comp f, T\inr\comp g]$.
\begin{definition}[Strong Separation]\label{def:strong-sep}
A separated abstract HO-GSOS law $(\rho^\val,\rho^\com,\dl)$ is \emph{strongly separated}
if the following diagram commutes:
\begin{equation}\label{eq:str_sep}
\begin{tikzcd}[column sep=-4em, row sep=normal]
&\Theta(X\times D(X,Y),X\times TY,X)
	\ar[dl,"\theta"']
	\ar[dr,"\theta"]&
\\
\Sigma_\com(X\times D(X,Y) + X\times TY,X)
	\ar[dd,"\iso"']
& &
\Sigma_\com(X\times D(X,Y) + X\times TY,X)
	\dar["{\Sigma_\com(\eta\comp\fst\klplus\snd,\inl)}"]
\\
& &
\Sigma_\com(T(X+Y),X+Y)
	\dar["\dl"]
\\
\Sigma_\com(X\times (D(X,Y)+TY),X)
	\ar[dr,"\rho^\com"']
& &
T\Sigma_\com(X+Y,X+Y)
	\ar[dl,"{T(\iota^\com\comp\Sigma_\com(\eta,\eta))}"]
\\
& T\Sigma^\star (X+Y) &
\end{tikzcd}
\end{equation}
\end{definition}
If $\Sigma_\com$ is a coproduct $\coprod_{i\in I}\Sigma_\com^i$, it suffices
to verify~\eqref{eq:str_sep} with $\Sigma_\com\ass\Sigma_\com^i$ for every $i$.
Intuitively, the left path of the diagram (from top to bottom) corresponds to
the general form of a rule, as represented by~$\rho^\com$, while the right path
specifies the required format. The commutativity of the diagram thus imposes a
constraint on~$\rho^\com$. Precomposition with~$\theta$ ensures that this constraint 
becomes effective only for rules with at least one premise from the computation 
part of the behaviour.

One could directly verify that strong separation holds for \xCL. It is instructive 
though to spell out~\eqref{eq:str_sep} for a larger class of examples, of which \xCL
is a member. 
\begin{remark}\label{rem:total-sep}
Let us spell out the strong separation condition in the total deterministic case, i.e.\
when $T=\Id$ and $\dl = \id$. The diagram \eqref{eq:str_sep} then simplifies as 
follows:
\begin{equation*}
\begin{tikzcd}[column sep=-4em, row sep=normal]
&\Theta(X\times D(X,Y),X\times Y,X)
	\ar[dl,"\theta"']
	\ar[dr,"\theta"]&
\\
\Sigma_\com(X\times D(X,Y) + X\times Y,X)
	\ar[d,"\iso"']
& &
\Sigma_\com(X\times D(X,Y) + X\times Y,X)
	\dar["{\Sigma_\com(\fst+\snd,\inl)}"]
\\[1ex]
\Sigma_\com(X\times (D(X,Y) + Y),X)
	\ar[dr,"\rho^\com"']
& &
\Sigma_\com(X+Y,X+Y)
	\ar[dl,"{\iota^\com\comp\Sigma_\com(\eta,\eta)}"]
\\
& \Sigma^\star (X+Y) &
\end{tikzcd}
\end{equation*}
Furthermore, let us interpret this diagram in the category of sets. Suppose 
that~$\Sigma_\com(X,Y)$ contains $F(X,Y)={X^n\times Y^m}$ as a summand for some natural 
numbers~$n$ and $m$. That is, $\Sigma_\com$ contains a computation former, say $f$,
whose~$n$ first arguments are strict and whose remaining~$m$ arguments are lazy.
As in the case of \xCL (\autoref{exa:xcl}), let $D(X,Y) = Y^X$. The restriction 
of $\rho^\com$ to~$F$ corresponds to the rules that describe the 
behaviour of terms of the form $f(x_1,\ldots,x_n,y_1,\ldots,y_m)$. The strong separation
condition requires the following: if the rule has at least one premise of the form 
$x_k\to x_k'$ then the conclusion of the rule must be of the form
\begin{align*}
f(x_1,\ldots,x_n,y_1,\ldots,y_m)\to f(x_1',\ldots,x_n',y_1,\ldots,y_m)
\end{align*}
where either $x_i\to x_i'$ occurs in the premise, or else, the premise contains a labeled 
transition for $x_i$, in which case $x_i'=x_i$.
\exend
\end{remark}
\begin{remark}
Our strong separation condition~\eqref{eq:str_sep} is a reminiscent of cool GSOS 
formats by~\mbox{\citet{Bloom95}} and~\citet{Glabbeek11} in the context 
of process algebra. These formats require that certain operations come with enough 
\emph{patience rules}, which are rules for the form
\begin{align*}
\frac{x_k\to x_k'}{f(x_1,\ldots,x_n)\to f(x_1,\ldots,x_{k-1},x_k',x_{k+1},\ldots,x_n)}
\end{align*}
This is needed to ensure that the semantics is sufficiently transparent to unlabeled 
transitions (or $\tau$-transitions in op.cit.). Originally, this enabled congruence properties of weak notions of process 
equivalence (such as weak bisimilarity). In our context, similar condition is needed to ensure 
that a small-step semantics can have a sensible a big-step semantics counterpart.
\end{remark}
\begin{example}
It is easy to see that \xCL satisfies strong separation: in view of \autoref{rem:total-sep},
the only relevant operation is application and the only relevant rule is
\begin{align*}
\infer{ts\to t's}{t\to t'}
\end{align*} 
Clearly, it has the requisite form.\exend
\end{example}
\begin{example}
Revisiting \autoref{exa:non-sep}, note that the strong separation condition is 
violated by the third rule.
As we argued, for the present small-step semantics we cannot define big-step semantics satisfying~\eqref{eq:sem-eq}. Thus,
in this case, strong separation effectively rules out an undesired example. \exend
\end{example}
As we see latter, the equivalence between the small-step and the big-step semantics
requires an $\omega$-continuous monad $\BBT$, while the identity monad is not $\omega$-continuous (Kleisli hom-sets do not have least elements).
We thus may need to adjoin a (separated) abstract higher-order GSOS law over a 
monad $\BBT$ to a given (separated) abstract higher-order GSOS law over the identity monad.
\begin{proposition}\label{pro:transfer}
Given families of morphisms
\begin{align*}
\rho_{X}^\val\c&\Sigma_\val X \to D(X,\Sigma^\star X),
\\*
\rho_{X,Y}^\com\c&\Sigma_\com(X\times (D(X,Y)+Y),X)\to \Sigma^\star (X+Y),
\end{align*}
natural in $Y$ and dinatural in $X$, a strong monad $\BBT$ and a distributive law $\dl_{X,Y}\c\Sigma_\com(TX,Y)\to T\Sigma_\com(X,Y)$, 
let $\hat\rho^\com$ be the family of morphisms, whose components are the compositions 
\begin{equation*}
\begin{tikzcd}[column sep=4ex,row sep=4ex]
\Sigma_\com(X\times B(X,Y),X)
	\ar[rr,"{\Sigma_\com(\id\times(\eta\klplus\id),\id)}"]
& &[-6ex]
\Sigma_\com(X\times T(D(X,Y)+Y),X)
	\ar[dll,"{\Sigma_\com(\tau,\id)}"',bend right=4,pos=.8]\\
\Sigma_\com(T(X\times (D(X,Y)+Y)),X)
	\rar[r,"\dl"]
& T\Sigma_\com(X\times (D(X,Y)+Y),X)
	\rar["T\rho^\com"] &
T\Sigma^\star (X+Y)
\end{tikzcd}
\end{equation*}
Suppose that the morphisms
\begin{align*}
\Theta(TX,TY,Z)\xto{\theta} \Sigma_\com(TX+TY,Z)\xto{\Sigma_\com(\id\klplus\id,\id)} \Sigma_\com(T(X+Y),Z)\xto{\dl} T\Sigma_\com(X+Y,Z)
\end{align*}
factor through $T\theta\c T\Theta(X,Y,Z)\to T\Sigma_\com(X+Y,Z)$.
Then if $(\rho^\val,\rho^\com,\id)$ is a strongly separated abstract higher-order GSOS law, then 
so is $(\rho^\val,\hat\rho^\com,\dl)$ .
\end{proposition}
A typical case for using \autoref{pro:transfer} is extending a deterministic law 
to a nondeterministic one, using the powerset monad. 
\begin{remark}\label{ref:transfer}
As in \autoref{rem:poly-theta}, consider $\Sigma_\com(X,Y) = X^n \times Y^m$, 
now on the category of sets. Let us choose $\Theta(X,Y,Z)\subseteq\Sigma_\com(X+Y,Z)$ to be $\{(\bar v,\bar z)\in (X+Y)^n\times Z^m\mid \exists i.\,\exists y\in Y.\,v_i = \inr y \}$
(so, $\theta$ is a set inclusion). Let $\BBT$ be the powerset functor $\PSet$, and assume 
the standard distributive law~\cite{HasuoJacobsEtAl07} $\dl_{X,Y}((A_i)_{i\in\{1,\ldots,n\}},\bar{y})=\{(\bar{x},\bar{y})\mid x_i\in A_i\}$.
Then, to show that $\dl \comp \Sigma_\com(\id \klplus \id , \id) \comp \theta$
factors through~$\PSet\theta$ is to show that for any $\bar V\in (\PSet X+\PSet Y)^n$ and $\bar z\in Z^m$
such that at least one $V_i$ is of the form $\inr A$, if $(\bar v,\bar z')\in (\dl \comp \Sigma_\com(\id \klplus \id , \id))(\bar V,\bar z)$
then also $v_i$ is of the form $\inr y$. Let us show it by contradiction: suppose that 
$(\bar v,\bar z')\in (\dl \comp \Sigma_\com(\id \klplus \id , \id))(\bar V,\bar z)$ but $v_i=\inl x$ for some $x\in X$.
Then for some %
$(\bar W,\bar z')\in (\Sigma_\com(\id \klplus \id , \id))(\bar V,\bar z)$, $W_i$ contains $\inl x$,
contradicting to the assumption that~$V_i=\inr A$.
\end{remark}

\section{Abstract Big-Step SOS}%

The natural transformations~\eqref{eq:separated1} and~\eqref{eq:separated2} model small-step 
semantics, and thus involve small-step transitions $\to$ and possibly others, which are binary 
relations between programs, i.e.\ closed $\Sigma$-terms. Contrastingly, the big-step
operational semantics involves judgements of the form $t\Dar v$ where~$t$ is a program
and $v$ is a value. The expected big-step operational semantics for the extended 
combinatory logic (\autoref{exa:xcl}) is provided in~\autoref{fig:big-ski}. 

The combinatory logic example suggests three natural questions:
\begin{enumerate}
  \item How to define a general notion of big-step semantics?
  \item How to generally produce a big-step specification from a small-step specification? 
  \item How to prove the equivalence of the big-step and the small-step semantics?
\end{enumerate}
We deal with the first and the second question in this section, and with the third one in the next~one.

\begin{definition}[Abstract Big-Step SOS]\label{def:big-step}
Given two functors $\Sigma_\val\c\BC\to\BC$, $\Sigma_\com\c\BC\times\BC\to\BC$, 
a strong monad $\BBT$ on $\BC$, an \emph{abstract big-step SOS} over these data is a natural transformation 
\begin{align}\label{eq:big-step}
	\xi\c\Sigma_\com(\Sigma_\val X,X)\to T(\Sigma^\star X)
\end{align}
where $\Sigma=\Sigma_\val+\Sigma_\com\Delta$.
\end{definition}
Assuming that $T=\Id$, \autoref{def:big-step} for one thing encodes the following 
concrete big-step operational semantics rules:
\begin{align*}
\infer{g(x_1,\ldots,x_n)\Dar g(x_1,\ldots,x_n)}{}
\end{align*}
for all $n$-ary symbols $g$ from $\Sigma_\val$. These rules are fully determined by the notion of value, i.e.\ by 
the component $\Sigma_\val$ of the partitioning of $\Sigma$ into $\Sigma_\val$ and $\Sigma_\com$. The transformation~\eqref{eq:big-step}
additionally encodes the rules of the form
\begin{align}\label{eq:bs-rule}
\frac{x_1\Dar g_1(x_1^1,\ldots,x_{n_1}^1)~~\ldots~~x_k\Dar g_k(x_1^k,\ldots,x_{n_k}^k) \qquad t\Dar v}{f(x_1,\ldots,x_n)\Dar v}
\end{align}
where $f$ is a computation former whose strict arguments are precisely the first $k$,~$t$ is a term whose free variables are in $\{x_1^1,\ldots,x_{n_1}^1,\ldots, x_1^k,\ldots,x_{n_k}^k,x_{k+1},\ldots,x_n\}$,
and~$v$ is a fresh variable referring to the result of evaluation of $t$. Specifications 
build from these two types of rules determine a natural transformation~\eqref{eq:big-step}
with $T=\Id$, and hence also such transformations with arbitrary $\BBT$ by composition 
with the monad unit $\eta$.

Given an abstract higher-order separated GSOS law $(\rho^\val,\rho^\com,\dl)$, we 
define~\eqref{eq:big-step} via:
\begin{equation}\label{eq:xi-constr}
\begin{tikzcd}[column sep=20ex,row sep=5ex]
\Sigma_\com(\Sigma_\val X,X)
	\rar{\Sigma_\com(\Sigma_\val\eta,\eta)}
&
\Sigma_\com(\Sigma_\val(\Sigmas X),\Sigmas X)
	\ar[ld,"{\Sigma_\com(\brks{\iota^\val, D(\id, \mmul) \comp  \rho^\val},\id)}"',pos=.9,bend right=7,start anchor={[yshift=-3pt]west}]
\\ \Sigma_\com(\Sigmas X\times D(\Sigma^\star X,\Sigma^\star X),\Sigmas X)
	\rar["{T \nabla^{\sharp}\comp\rho^{\com\val}}"]
&
T(\Sigma^\star X)
\end{tikzcd}
\end{equation}
This is indeed the way, the rules in \autoref{fig:big-ski} are obtained from the 
rules in \autoref{fig:rules}.\sgnote{Maybe add an example.} %

Using~\eqref{eq:xi-constr} in derivations requires instantiating $X$ with $\mS$ and 
flattening~$\Sigmas(\mS)$ to $\mS$. This results in a significant simplification.
\begin{lemma}\label{lem:xi-mS}
The following diagram commutes:
\begin{equation*}
\begin{tikzcd}[column sep=1.5em, row sep=normal]
\Sigma_\com(\mSv,\mS)
	\ar[rr,"\xi"] 
	\dar["{\Sigma_\com(\brks{\iota^\val,\gamma^\val},\id)}"']
&[.5em]&
T(\Sigmas\mS)
	\dar["T\mmul"]               
\\
\Sigma_\com(\mS\times D(\mS,\mS),\mS)
	\ar[r,"\rho^{\com\val}"] 
&
T\Sigmas(\mS+\mS)
	\ar[r,"T\nabla^\klstar"]
&
T(\mS)               
\end{tikzcd}
\end{equation*}
\end{lemma}
The proof of this lemma uses the (di-)naturality assumption on~\eqref{eq:separated2}.
\begin{example}
Assuming that $\BBT$ is the identity monad, for the extended combinatory logic (\autoref{exa:xcl})
we obtain the following assignments:
\begin{align*}
\xi (\tcomp(I, r)) 			=&\; r\qquad	&\xi (\tcomp(K, r))  =&\; K'(r)	&\qquad \xi (\tcomp(S, r))	=&\; S'(r)\\ 
\xi (\tcomp(K'(t), r)) 		=&\; t		&\xi (\tcomp(S'(t), r))	=&\; S''(t, r)\\
\xi (\tcomp(S''(s,t), r)) 	=&\; \makebox[0pt][l]{$\tcomp(\tcomp(s, r),\tcomp(t, r))$}\\[-8ex]
\end{align*}
\exend\end{example}
The power of abstract (higher-order) GSOS semantics is not only in that the derivation rules
can be captured by a single (di-)natural transformation~\eqref{eq:ho-gsos-law},
but also in that this transformation allows one to execute this semantics via the 
ensuing notion of operational model $\gamma$, determined by~\eqref{eq:op-mod} 
just as abstractly. Analogously, we need to define how big-step operational semantics is 
executed, i.e.\ how the process of deriving the judgements $t\Dar v$ is modelled, 
by providing a big-step counterpart of the HO-GSOS operational model. 
We seek to define a morphism
\begin{align}\label{eq:bss}
\hat\zeta\c\mS\to T(\mSv),
\end{align}
which is meant to send a term $t$ to a value $v$, possibly triggering a side-effect
(in the simplest case, partiality). In general, $\BBT$ must be $\omega$-continuous.

Let us explain briefly the role of the monad. The rule~\eqref{eq:bs-rule} involves premises that try to evaluate $x_1,\ldots,x_k$,
which are all structurally smaller than $f(x_1,\ldots,x_n)$ from the conclusion. 
Hence, any derivation w.r.t.\ to these premises is necessarily well-founded. However, 
the premise $t\Dar v$ involves $t$, which need not be structurally smaller than
$f(x_1,\ldots,x_n)$. For a given term of the form $f(t_1,\ldots,t_n)$, we thus 
cannot generally decide if $f(t_1,\ldots,t_n)\Dar v$ is derivable for some $v$.
This phenomenon is not an artefact of the format~\eqref{eq:bs-rule}, but an inherent 
feature of the big-step semantics.
Consider for example the terms $\Omega_{k} = (I^k(SII))(I^k(SII))$ of \xCL (where $U^k$ refers
to the term $U(\ldots U(U U)\ldots)$ with $U$ repeated $k$ times). We have
\begin{align*}
\Omega_k\to^\star (SII)(I^k(SII))
\to (S'(I)I)(I^k(SII))\to (S''(I,I))(I^k(SII))\to %
\Omega_{k+1}. 
\end{align*}
Even though every $\Omega_k$ has an outgoing unlabeled transition, no value
can be reached from any of the $\Omega_k$. In the big-step style this means that 
$\Omega_k\Dar v$ is not derivable for any $v$. Let us illustrate that with the 
following derivation fragment:
\begin{align*}
\infer{\Omega_{k}\Dar v}{\infer{I^k(SII)\Dar S''(I,I)}{\infer{\vdots}{\infer{SII\Dar S''(I,I)}{\infer{SI\Dar S'(I)}{\infer{S\Dar S}{}}}}}~~&~~\infer{\Omega_{k+1}\Dar v}{?}}
\end{align*}
The rightmost premise of the rule would require a finite derivation whose size could 
not be smaller than that for the original goal $\Omega_{k}\Dar v$; hence, neither derivation 
can exist, i.e.\ $\Omega_{k}\Dar v$ is not derivable. %

Assuming that $\BBT$ is $\omega$-continuous and an $\omega$-continues distributive 
law~\eqref{eq:separated3}, we introduce~\eqref{eq:bss} as the least fixpoint 
\begin{align}\label{eq:bs-xi}
	\hat\zeta = \mu f.\, [\eta,f^\klstar\comp T\mu\comp\xi^\klstar\comp\dl\comp\Sigma_\com(f,\id)]\comp\iota^\mone
\end{align} 
where the composition $f^\klstar\comp T\mu\comp\xi^\klstar\comp\dl\comp\Sigma_\com(f,\id)$ spells out as follows:
\begin{equation*}
\begin{tikzcd}[column sep=5em, row sep=normal]
\Sigma_\com(\mS,\mS)
	\rar["{\Sigma_\com(f,\id)}"] &
\Sigma_\com(T\mSv,\mS)
	\rar["\dl"] &
T\Sigma_\com(\mSv,\mS)
	\ar[dll,"\xi^\klstar"',bend right=4,pos=.8]
\\
T(\Sigmas\mS)
	\rar["T\mu",pos=.6] 
&
T(\mS) 
	\rar["f^\klstar"]
&
T(\mSv)
\end{tikzcd}
\end{equation*}
Let $\zeta = \hat\zeta\comp\iota^\com\c\mSc\to T(\mSv)$. By definition, $\hat\zeta=[\hat\zeta\comp\iota^\val,\hat\zeta\comp\iota^\com]\comp\iota^\mone=[\eta,\zeta]\comp\iota^\mone$.

\section{Equivalence of Small-Step and Big-Step, Abstractly}%

In this section we assume that the monad $\BBT$ is $\omega$-continuous. 
We then abstractly define multi-step transitions $t\to^\star v$ from computations to values as the least fixpoint
\begin{align}\label{eq:multi}
\beta = (T\iota^\mone\comp\gamma^\com)^\istar =\mu f.\, [\eta,f]^\klstar\comp T\iota^\mone\comp\gamma^\com \c\mSc\to T(\mSv).
\end{align}
In other words, $\beta$ is the least solution of the equation $\beta = [\eta,\beta]^\klstar\comp T\iota^\mone\comp\gamma^\com$. 
Let moreover $\hat\beta = [\eta,\beta]\comp\iota^\mone\c\mS\to T(\mSv)$.

Suppose that $D(X,Y) = Y^X$ and that $\Sigma$ is some algebraic signature on~$\Set$. 
Let $f$ be some $(n+m)$-ary computation former with precisely~$n$ strict first arguments.
Suppose that $t_1\to^\star v_1,\ldots,t_n\to^\star v_n$ and that ${f(v_1,\ldots,v_n,t_{n+1},\ldots,t_m)\to t\to^\star v}$ 
for some values $v_1,\ldots,v_n,v$. We then expect that $f(t_1,\ldots,t_{n+m})\to^\star v$.
The following lemma captures this fact abstractly.
\begin{lemma}\label{lem:beta-delta}
Let $(\rho^\val,\rho^\com,\dl)$ be a strongly separated abstract higher-order GSOS law,
such that the law $\dl$ is $\omega$-continuous. 
Then 
\begin{align}\label{eq:beta-ineq}
\hat\beta^\klstar\comp T\nabla^\klstar\comp(\rho^{\com\val})^\klstar\comp\dl\comp \Sigma_\com(T\brks{\iota^\val,\gamma^\val}\comp\hat\beta,\id)\appr\beta.
\end{align}
\end{lemma}
\begin{proof}[Proof Sketch]
The key step is proving that the following diagram commutes:
\begin{equation}
\begin{tikzcd}[column sep=normal, row sep=normal]
\mSc
	\rar["\beta"]
	\dar["{\Sigma_\com(\hat\gamma^\com,\id)}"'] 
&
T(\mSv)
\\               
\Sigma_\com(T\mS,\mS)
	\rar["\dl"] 
&
T(\mSc)
	\uar["\beta^\klstar"']
\end{tikzcd}
\end{equation}
To that end we use the fact that, by \autoref{ass:compl}, $\mSc$ is a coproduct of $\Sigma_\com(\mSv,\mS)$ and $\Theta(\mSv,\mSc,\mS)$.
The diagram then falls into two equations:
\begin{align*}
\beta^\klstar\comp\dl\comp\Sigma_\com(\hat\gamma^\com,\id)\comp\Sigma_\com(\iota^\val,\id) =&\; \beta\comp\Sigma_\com(\iota^\val,\id),\\
\beta^\klstar\comp\dl\comp\Sigma_\com(\hat\gamma^\com,\id)\comp\Sigma_\com(\iota,\id)\comp\theta =&\; \beta\comp\Sigma_\com(\iota,\id)\comp\theta.
\end{align*}
The first is obtained by unfolding definitions and using the fact that $\dl$ is a distributive law.
The second one relies on~\eqref{eq:str_sep}. Using the 
diagram, the goal is obtained by induction, using the fact that
(by Kleene's fixpoint theorem) $\hat\beta = \bigjoin_n\hat\beta_n$ where 
$\hat\beta_0 = [\eta,\bot]\comp\iota^\mone$, $\hat\beta_{n+1} = \hat\beta_{n}^\klstar\comp\hat\gamma^\com$.
\end{proof}

Another property that we need to abstract is that $t\to s$ together with $s\Dar v$
entails $t\Dar v$.
\begin{lemma}\label{lem:big-and-small}
Let $(\rho^\val,\rho^\com,\dl)$ be an strongly separated abstract HO-GSOS with 
$\omega$-continuous~$\dl$ and let~$\xi$
be defined by~\eqref{eq:xi-constr}. 
Then $\hat\zeta^\klstar\comp\gamma^\com\appr\zeta$.
\end{lemma}
\begin{proof}[Proof Sketch]
The diagram~\eqref{eq:op-mod-c} identifies
$\gamma^\com$ as the unique fixpoint of the map $f\mto T\nabla^\klstar\comp\rho^\com\comp\Sigma_\com(\brks{\id,(\gamma^\val+ f)\comp\iota^\mone},\id)$,
hence, also as the least fixpoint. Thus $\gamma^\com = \bigjoin_n \gamma^\com_n$ where $\gamma^\com_0 = \bot$ and $\gamma^\com_{n+1} = T\nabla^\klstar\comp\rho^\com\comp\Sigma_\com(\brks{\id,(\gamma^\val+\gamma_n^\com)\comp\iota^\mone},\id)$, and it suffices to show $\hat\zeta^\klstar\comp\gamma^\com_n\appr\zeta$
for all $n$. The induction base is obvious. The induction step, by \autoref{ass:compl}, we reduce to two subgoals:
\begin{align*}
\hat\zeta^\klstar\comp\gamma^\com_{n+1}\comp\Sigma_\com(\iota^\val,\id) 		\appr&\; \zeta\comp\Sigma_\com(\iota^\val,\id),\\*
\hat\zeta^\klstar\comp\gamma^\com_{n+1}\comp\Sigma_\com(\iota,\id)\comp\theta   \appr&\; \zeta\comp\Sigma_\com(\iota,\id)\comp\theta.
\end{align*}
The first one is easy to obtain by unfolding definitions. The second one is a result of a somewhat tedious 
calculation, using the induction hypothesis and the strong separation assumption.
\end{proof}

We proceed with formalizing and proving an abstract version of the 
equivalence~\eqref{eq:sem-eq} between the small-step and the big-step semantics
for strongly separated abstract HO-GSOS laws. To that end, we first rewrite the formula~\eqref{eq:bs-xi} in the setting 
when $\xi$ comes from~\eqref{eq:xi-constr}.

\begin{proposition}\label{lem:bs-fix}
Let $(\rho^\val,\rho^\com,\dl)$ be an abstract higher-order separated GSOS with $\omega$-continuous 
$\dl$ and let~$\xi$ be defined by~\eqref{eq:xi-constr}. Then 
\begin{align}\label{eq:bs-xi-fix}
	\hat\zeta = \mu f.\,[\eta, f^\klstar\comp T\nabla^\klstar\comp(\rho^{\com\val})^\klstar\comp\dl\comp \Sigma_\com(T\brks{\iota^\val,\gamma^\val}\comp f,\id)]\comp\iota^\mone.
\end{align}
\end{proposition}
Finally, we prove our main result.
\begin{theorem}\label{the:main}
Let $(\rho^\val,\rho^\com,\dl)$ be a strongly separated abstract higher-order  GSOS with 
$\omega$-continuous~$\dl$ and let~$\xi$
be defined by~\eqref{eq:xi-constr}. 
Then $\hat\zeta = \hat\beta$.
\end{theorem}
\begin{proof}
We proceed by proving mutual inequality. 

$\hat\zeta\appr\hat\beta$. By \autoref{lem:bs-fix}, $\hat\zeta$ is the least pre-fixpoint of 
\eqref{eq:bs-xi-fix}. By \autoref{lem:beta-delta}, $\hat\beta$ is a pre-fixpoint of the 
same function. Hence, indeed, $\hat\zeta\appr\hat\beta$.

$\hat\beta\appr\hat\zeta$. The goal is equivalent to $[\eta,\beta]\comp\iota^\mone\appr[\eta,\zeta]\comp\iota^\mone$.
Proving this is equivalent to proving that $\beta\appr\zeta$, which will then be our new goal. By definition,~$\beta$ is the least pre-fixpoint of the map
$f\mto [\eta,f]^\klstar\comp T\iota^\mone\comp\gamma^\com$. Hence, it suffices to
prove that $\zeta$ is a pre-fixpoint of the same map, i.e.\ the inequality
$[\eta,\zeta]^\klstar\comp T\iota^\mone\comp\gamma^\com\appr\zeta$,~i.e.\ $\hat\zeta^\klstar\comp\gamma^\com\appr\zeta$,
which we have by \autoref{lem:big-and-small}.
\end{proof}

\section{Case Studies}\label{sec:exa}
We proceed to consider various aspects of programming languages from the perspective
of our framework, and demonstrate in detail, how it can cope with them.  
\subsection{Types}
In this section, we stick to a typed version of $\xCL$, previously dubbed $\xTCL$~\cite{GoncharovSantamariaEtAl24}. 
The transition from $\xCL$ to $\xTCL$ demonstrates the strength of our categorical modeling: from
$\Set$ as the ambient category $\BC$ we simply switch to the category of sorted sets $\Set^\Ty$,
where the sorts $\Ty$ are indeed the conventional types of typed combinatory logic~\cite{HindleySeldin08}. The remaining 
ingredients of the framework are then upgraded as expected, which we detail further below. 

Assuming a postulated set of basic sorts $\CS$, the set of types $\Ty$ is defined by the grammar:
\begin{flalign*}
&&\Ty \Coloneqq (\tau\in\CS)  \mid  \arty{\Ty}{\Ty}. && 
\end{flalign*}
More concretely, an object $A$ 
of~$\BC$ is a family of sets $(A_\tau)_{\tau \in \Ty}$ and a morphism $f \c A \to B$ in $\BC$ is a family of functions
$(f_\tau \c A_\tau \to B_\tau)_{\tau\in\Ty}$. Morphism composition $g \comp f$ is defined
as $(g_\tau \comp f_\tau)_{\tau \in \Ty}$. Note that $\BC$ is a presheaf topos over~$\Ty$ as a discrete category, 
in particular, limits and colimits in $\BC$ are computed pointwise.  
We introduce $\Sigma_\val$ and $\Sigma_\com$ for $\xTCL$ as follows:
\begin{align*}
	\Sigma_\val (X)_\tau =&\; \coprod\nolimits_{f\c\tau_1,\ldots, \tau_n \to \tau\in Ops} \prod\nolimits_{i=1}^n X_{\tau_i}
	&&\Sigma_\com (X,Y)_\tau= \{\tcomp_{\tau',\tau}\}\times (X_{\arty{\tau'}{\tau}} \times Y_{\tau'}),
\end{align*}
where $Ops=\{I_{\tau_1}, K_{\tau_1,\tau_2}, K'_{\tau_1,\tau_2}, S_{\tau_1,\tau_2,\tau_3}, S'_{\tau_1,\tau_2,\tau_3}, S''_{\tau_1,\tau_2,\tau_3}\}$
is the set of operations, constituting $\Sigma_\val$, and typed in the expected way. For example, the type of~$K'_{\tau_1,\tau_2}$ is $\tau_1 \to (\arty{\tau_2}{\tau_1})$. 
Analogously, we view~$\Sigma_\com$ as a signature of one binary operation $\tcomp_{\tau_1,\tau_2}$ of type $(\arty{\tau_1}{\tau_2}), \tau_1 \to \tau_2$.

To define the behaviour functor $B$, let $\BBT$ be the identity monad, and let $D$ be as follows: 
\begin{align*}
D(X,Y)_{\tau}=\{\}\text{\quad for $\tau\in\CS$}, \qquad D(X,Y)_{\arty{\tau_1}{\tau_2}}=Y_{\tau_2}^{X_{\tau_1}}
\end{align*}
and then $B(X,Y)_\tau =\fpc_\tau$ if $\tau\in\CS$ and $B(X,Y)_{\arty{\tau_1}{\tau_2}} = Y_{\arty{\tau_1}{\tau_2}}+ Y_{\tau_2}^{X_{\tau_1}}$. %
Next, we define $\rho^\val$ and $\rho^\com$ essentially like in \autoref{exa:xcl}, modulo 
adding the typing information. For example, the clause for~$S''$ in the definition of 
$\rho^\val_{\tau}\c\Sigma_\val (X)_{\tau} \to D(X,\Sigmas X)_{\tau}$ becomes
\begin{align*}
	\rho^\val_{\tau}(S''_{\tau_1,\tau_2,\tau}(t,s))(r)=&\; 
	\tcomp_{\tau_2,\tau} (\tcomp_{\tau_1,\arty{\tau_2}{\tau}} (t,r) ,\tcomp_{\tau_1,\tau_2}(s,r)).
\end{align*}
The definition of $\rho^\com_{\tau}\c\Sigma_\com (X \times B(X,Y), X)_{\tau} \to \Sigma^\star(X+Y)_{\tau}$ is as follows:
\begin{align*}
	\rho^\com_\tau (\tcomp_{\tau_1,\tau_2} ((s, s'),t)) =&\; \tcomp_{\tau_1,\tau_2} (s',t) 	&&\text{if}\; s' \in Y_{\arty{\tau_1}{\tau_2}}\\*
	\rho^\com_\tau (\tcomp_{\tau_1,\tau_2} ((s, f),t)) =&\; f(t) 														&&\text{if}\; f \in Y_{\tau_2}^{X_{\tau_1}}
\end{align*}
Finally, we obtain $\xi_\tau \c \Sigma_\com (\Sigma_\val X , X)_\tau \to \Sigma^\star (X)_\tau$ by instantiating
the general definition~\eqref{eq:xi-constr}. For example, we have
	$\xi_\tau(\tcomp_{\tau_1,\tau} (S''_{\tau_1,\tau_2,\tau}(t,s),r\c\tau_1))=
	\tcomp_{\tau_2,\tau}(\tcomp_{\tau_1,\arty{\tau_2}{\tau}}(t,r),\tcomp_{\tau_1,\tau_2}(s,r))$.
\subsection{Recursion and Conditionals}

In this section, we complement $\xTCL$ with recursion and control
in the style of PCF~\cite{Mitchell96}.
We assume the sort $\bool \in \mathcal{S}$ of Booleans. Then we add the following operators to the signature:
\begin{align*}
	\fpc_\tau \c& \arty{(\arty{\tau}{\tau})}{\tau}&
	 \casec_\tau \c& \bool , \tau , \tau \to \tau&
	\true,
	\false \c& \bool
\end{align*}
More concretely, this amounts to defining $\Sigma_\val$ with the same formula 
from $Ops$, which is now the set $\{I_{\tau_1}, K_{\tau_1,\tau_2}, K'_{\tau_1,\tau_2}, S_{\tau_1,\tau_2,\tau_3}, S'_{\tau_1,\tau_2,\tau_3}, S''_{\tau_1,\tau_2,\tau_3}, \fpc_{\tau_1}\}$,
and to defining $\Sigma_\com$ as follows:
\begin{gather*}
\Sigma_\com (X,Y)_\tau= \{\tcomp_{\tau',\tau}\}\times (X_{\arty{\tau'}{\tau}} \times Y_{\tau'}) \;+\; \{\casec_{\tau}\}\times (X_{\bool} \times Y_\tau \times Y_\tau).
\end{gather*}
The left summand corresponds to the application operator $\tcomp_{\tau_1,\tau_2}$, as before, while the 
right summand corresponds to the conditional statement $\casec_\tau$.
This indicates that the first argument of~$\casec_\tau$ is strict and the others 
are lazy.
The new operators are subject to the following small-step specification:
\begin{gather*}
\frac*{\fpc_\tau \xto{t}t (\fpc_\tau t)}{}\qquad\quad
\frac*{ \true \dar_\true}{} \qquad\quad
\frac*{ \false \dar_\false}{} \\[1ex]
\frac*{ \casec_\tau (b,s,t) \to  \casec_\tau (b',s,t)}{b \to b'}\qquad\quad
\frac*{ \casec_\tau (b,s,t) \to s }{b \dar_\true}\qquad\quad
\frac*{ \casec_\tau (b,s,t) \to t }{b \dar_\false}
\end{gather*}

\medskip\noindent
where we use a new type of judgements $b\dar_\true$ and $b\dar_\false$, indicating 
that $b=\true$ and $b=\false$ correspondingly.
Note that the if-then-else statement suffices to program the standard logical 
connectives: $\neg b = \casec_{\bool}(b,\false,\true)$, $b\land b' = \casec_{\bool}(b,b',\false)$, 
$b\lor b' = \casec_{\bool}(b,\true,b')$.
To address the above rules we modify~$D$ as follows:
\begin{align*}
	D(X,Y)_{\arty{\tau_1}{\tau_2}}=Y_{\tau_2}^{X_{\tau_1}} && D(X,Y)_\tau = \begin{cases}\{\true,\false\} & \text{~~if~~}\tau=\bool\\ \{\}&\text{~~if~~}\tau\in\CS\setminus\{\bool\}\end{cases}
\end{align*}
Finally, we complement the definitions of $\rho^\val$ and $\rho^\com$ with the clauses
\begin{align*}
	\rho^\val_\tau(\fpc_\tau) (t) =&\; \tcomp_{\tau,\tau} (t , \tcomp_{\arty{\tau}{\tau}, \tau}(\fpc_\tau,t))\!\!&
	\rho^\com_\tau(\casec_\tau((b,b') , s , t))=&\;  \casec_\tau(b' , s , t)
	\\*
	\rho^\val_\bool(\true) =&\; \true&
	\rho^\com_\tau(\casec_\tau((b,\true) , s , t))=&\; s
	\\*
	\rho^\val_\bool(\false) =&\; \false&
	\rho^\com_\tau(\casec_\tau((b,\false) , s , t))=&\; t
\end{align*}
It is easy to see by~\autoref{rem:total-sep} that $(\rho^\val,\rho^\com)$ 
constitute a strongly separated abstract HO-GSOS law.
By applying~\eqref{eq:xi-constr} we obtain the following new clauses for $\xi$:
\begin{align*}
	\xi_\tau (\tcomp_{\arty{\tau}{\tau},\tau} (\fpc_\tau , t \c \arty{\tau}{\tau})) =&\; \tcomp_{\tau,\tau} (t , \tcomp_{\arty{\tau}{\tau}, \tau}(\fpc_\tau,t))\\
	\xi_\tau (\casec_\tau (\true,s,t))  =&\; s\\
	\xi_\tau (\casec_\tau (\false,s,t)) =&\; t
\end{align*}
Equivalently: the following big-step rules:
\begin{gather*}
	\infer{s t \Dar v}{s \Dar \fpc_\tau & t  (\fpc_\tau t) \Dar v}\qquad\quad
	\infer{\casec_\tau(b,s,t) \Dar v}{b \Dar \true & s \Dar v} \qquad\quad
	\infer{ \casec_\tau(b,s,t) \Dar v}{b \Dar \false & t \Dar v}
\end{gather*}

\subsection{Nondeterminism and Parallelism} \label{sec:non-d}
We proceed to illustrate how our modeling can cope with nondeterminism and parallelism
in higher-order setting by extending $\xCL$ suitably. The standard binary \textit{erratic choice} 
operator $\ndet$ \cite{DeliguoroPiperno95} can easily be specified with the small-step 
rules:
\begin{align}\label{eq:err-choice}
	\infer{t \ndet s \to t}{} \qquad\qquad
	\infer{t \ndet s \to s}{}
\end{align}
This amounts to updating the separable abstract HO-GSOS law in \autoref{exa:xcl}
as follows: $\Sigma_\com(X,Y) = \{\tcomp\}\times (X\times Y)+\{\oplus\}\times(Y\times Y)$, 
$\BBT$ is the powerset monad $\PSet$, and the clauses for $\rho^\com$ are as follows:
\begin{align*}
	\rho^\com(\tcomp((t,S),s))	=&\; \{\tcomp(t', s) \mid t' \in S\cap Y\} \cup \{f(s) \mid f\c X\to Y \in S\}\\
		\rho^\com(t \ndet s)=&\; \{t,s\}
\end{align*}
The rules in~\eqref{eq:err-choice} thus jointly correspond to a single clause 
for $\rho^\com$, stating that $\{t,s\}$ is the set of direct successors of $t\ndet s$.
The construction~\eqref{eq:xi-constr} produces an abstract big-step SOS, which can be 
represented with the rules:
\begin{align*}
	\infer{t \ndet s \Dar v}{t \Dar v}\qquad\qquad
	\infer{t \ndet s \Dar v}{s \Dar v}
\end{align*}
In order to specify the behaviour of a parallel composition operator we orient to
\emph{concurrent $\lambda$-calculus}~\cite{Dezani-CiancagliniEtAl98}. This is an extension of 
the (call-by-name) $\lambda$-calculus with the erratic choice operator, as above, with a (fair)
parallel composition operator $\paral$ and with a certain version of call-by-value evaluation.
The latter feature is independent of the others, and was added to express the well-known
\emph{parallel-or} operator. We will presently not include call-by-value evaluation, 
but consider it in dedicated~\autoref{sec:cbv} instead. Intuitively, in $s\paral t$,
$s$ and $t$ are reduced simultaneously if possible, and otherwise the one 
that can be reduced is reduced, while keeping the other one intact; eventually,
neither~$s$ nor~$t$ can be reduced, meaning that $s\paral t$ is morally a value. 
We capture the latter situation by arranging a reduction from $s\paral t$ to $s\paralv t$
where $\paralv$ is a new value former.

\begin{gather*}
	\frac{t \to t' \quad s \to s'}{t \paral s \to t' \paral s'}	\qquad
	\frac{t \xto{r} t' \quad s \to s'}{t \paral s \to t \paral s'}\qquad
	\frac{t \to t' \quad s \xto{r} s'}{t \paral s \to t' \paral s}\\[1ex]
	\frac{t \xto{r} t' \quad s \xto{r} s'}{t \paral s \to t \paralv s}\qquad
	\frac{}{t\paralv s \xto{r} tr\paral sr}
\end{gather*}
These rules yield the following new clauses for $\rho^\com$:
\begin{align*}
\rho^\com((t,T)\paral (s,S))	=&\; \{t'\paral s'\mid t'\in T,s'\in S\}&
\rho^\com((t,f)\paral (s,S))	=&\; \{t\paral s'\mid s'\in S\}\\
\rho^\com((t,T)\paral (s,g))	=&\; \{t'\paral s\mid t'\in T\}&
\rho^\com((t,f)\paral (s,g))	=&\; \{t\paralv s\}
\end{align*}
and the following new clause for $\rho^\val$:
\begin{align*}
\rho^\val(t\paralv s)(r) = (t\cdot r)\paral (s\cdot r).
\end{align*}
The new big-step rules are as follows:
\begin{gather*}
	\infer{s t\; \Dar   w}{s\; \Dar \; v \paralv u\qquad v t\paral ut\;\Dar w}\qquad\qquad 
	\infer{s \paral t \;\Dar v \paralv u}{s \Dar v \qquad t \Dar u & } 
\end{gather*}

\subsection{Call-by-Value}\label{sec:cbv}
Rules in \autoref{fig:rules} can be easily modified to capture the familiar call-by-value 
evaluation strategy:\sgnote{Maybe use letter codes to refer to values.}
\begin{gather*}
\frac{}{S\xto{t}S'(t)}
\qquad
\frac{}{S'(t)\xto{s}S''(t,s)}
\qquad
\frac{}{S''(t,s)\xto{r}(tr)(sr)}
\\[2ex]
\frac{}{K\xto{t}K'(t)}
\qquad
\frac{}{K'(t)\xto{s}t}
\qquad
\frac{}{I\xto{t}t}
\\[2ex]
\frac{t\to t'}{t s\to t's}~~ (a)
\qquad
\frac{t\xto{r} t'\quad s\to s'}{t s\to ts'}~~ (b)
\qquad
\frac{t\xto{s} t'\quad s\xto{r} s'}{t s\to t'}~~ (c) %
\end{gather*}
In this specification, we reduce $ts$ by first reducing $t$ (a), unless it expects an
input, i.e.\ $t$ is a value, in which case we reduce $s$ (b), unless also $s$ is a value; in the remaining 
case we evaluate $t$ on $s$ (c).
To turn this specification into a separated abstract HO-GSOS, 
we need to decide which arguments of the application operator are lazy 
and which are strict. In general, the behaviour of~$ts$ depends both on the 
behaviour of $t$ and of~$s$, hence both arguments must be strict. However, as explained
in \autoref{rem:total-sep}, for the above specification to be strongly separated, it must contain 
the rule 
\begin{align*}
\frac{t \to t' \quad s \to s'}{t \app s \to t' \app s'}~~(a1)
\end{align*}
which is not the case. We can amend the original specification by replacing 
(a) by (a1) together with
\begin{align*}
\frac{t \to t' \quad s \xto{r} s'}{t \app s \to t' \app s}~~(a2)
\end{align*}
This results in a specification to which we can apply \autoref{the:main} and 
obtain the equivalence of the small-step semantics with the following big-step semantics:
\begin{equation}\label{eq:bs-cbv}
\begin{gathered}
	\infer{v \Dar v}{} \qquad\qquad
	\infer{st\Dar v}{s \Dar I & t \Dar v} \qquad\qquad
	\infer{st\Dar v}{s \Dar K \quad t\Dar w\quad  K'(w)\Dar v}\\[2ex]
	\infer{st\Dar v}{s \Dar S \quad t\Dar w\quad S'(w)\Dar v}  \qquad\qquad
	\infer{st\Dar v}{s \Dar K'(r) \quad t\Dar w\quad r \Dar v} \\[2ex]
	\infer{st\Dar v}{s \Dar S'(r) \quad t\Dar w \quad S''(r,w)\Dar v}\qquad\quad
	\infer{st\Dar v}{s \Dar S''(r,q) \quad t\Dar w\quad (rw)(qw)\Dar v}
\end{gathered}
\end{equation}
Although, the variant of the small-step specification with (a) replaced by (a1) 
and (a2) makes perfect sense, the resulting multi-step relation $\to^\star$ is not the 
standard one: e.g.\ originally $(I\app I) \app (I \app I) \to^\star I \app (II)$, but 
in the modified system $(I\app I) \app (I \app I) \to I \app I$, and hence $(I\app I) \app (I \app I) \not\to^\star I \app (II)$.  
To model the standard multi-step relation, we replace the rules (a)--(c) with the 
following ones:
\begin{gather*}
	\frac{s \to s'}{s \app t\to s' \app t}\qquad
	\frac{s\xto{r} s'}{s\app t\to s \lcomp t}\qquad
	\frac{t \to t'}{s \lcomp t \to s \lcomp t'}\\[2ex]
	\frac{t\xto{r} t'}{s \lcomp t \to s \fcomp t}\qquad
	\frac{s \to s'}{s \fcomp t \to s' \fcomp t}\qquad
	\frac{s \xto{t} s'}{s\fcomp t \to s'}
\end{gather*}
where $\lcomp$ and $\fcomp$ are new auxiliary composition operators. The idea is to treat 
the original application operator (juxtaposition) and $\fcomp$ as strict in the first argument 
and lazy in the second, and $\lcomp$ as lazy in the first argument and strict in the second.
The original multi-step behaviour is thus recovered,
strong separation holds and the equivalent big-step specification takes the form:
\begin{gather*}
\infer{v \Dar v}{} \quad\qquad
\infer{st \Dar v}{s \Dar w & w\lcomp t \Dar v} \quad\qquad
\infer{s\lcomp t \Dar v}{t \Dar w & s\fcomp w \Dar v} \quad\qquad
\infer{s\fcomp t\Dar v}{s\Dar I & t \Dar v} \quad\qquad
\infer{s\fcomp t\Dar v}{s\Dar K \quad K'(t)\Dar v}
\\[2ex]
\infer{s\fcomp t\Dar v}{s\Dar S \quad S'(t)\Dar v}\quad\quad ~~ 
\infer{s\fcomp t\Dar v}{s\Dar K'(r) \quad r \Dar v}\quad\quad ~~
\infer{s\fcomp t\Dar v}{s\Dar S'(r) \quad S''(r,t)\Dar v}\quad\quad ~~
\infer{s\fcomp t\Dar v}{s\Dar S''(r,q) \quad (rt)(qt)\Dar v}
\end{gather*}
Observe that the only way to obtain $st\Dar v$ is by using the following derivation
\begin{align*}
\infer{st\Dar v}{\infer{s \Dar w}{\vdots} \quad &\quad  \infer{w\lcomp t \Dar v}{\infer{t \Dar u}{} \quad &\quad \infer{w\fcomp u \Dar v}{\vdots}}}
\end{align*}
where $\vdots$ refers to undischarged premises. By inspecting the six 
rules whose conclusions match $w\fcomp u \Dar v$, it is easy to see that for terms 
in the original signature (without~$\lcomp$ and $\fcomp$) we obtain the same big-step 
semantics as in~\eqref{eq:bs-cbv}. In summary,~\eqref{eq:bs-cbv} is equivalent to
the two variants of the small-step semantics: the one with the rules (a), (b), (c) and 
the one with the rules (a1), (a2), (b), (c). Hence, they are equivalent to each other
in the sense that $t\to^\star v$ in one system iff $t\to^\star v$ in the other system
(for any value $v$). 

Our use of the auxiliary operators $\lcomp$ and $\fcomp$ tactically helped us to 
satisfy our strong separation condition. However, such operations emerge independently
in the context of \emph{pretty-big-step semantics}~\cite{Chargueraud13}. In fact, 
one can say that what we 
generally call abstract big-step SOS, specialized to the pretty-big-step semantics in this case.

\section[Variable Binders and the λ-calculus]{Variable Binders and the $\lambda$-calculus}\label{sec:lam}
So far, we stuck to the extended combinatory logic as the core higher-order vehicle  
for language features of interest. Modeling languages with binders adds a
certain technical overhead to the problem of modeling the small-step and 
the big-step semantics categorically. In fact, the type profile in \eqref{eq:separated1} 
turns out to be too restrictive. In this section we show how to remedy this and 
to cover the $\lambda$-calculus. Similar issue arose recently in categorical 
treatment of logical predicates for languages with binders~\cite{GoncharovSantamariaEtAl24},
and required a refinement of abstract HO-GSOS, called $\lambda$-laws.

\subsection{Separable HO-GSOS for Languages with Binders}\label{sec:bind-gsos}

We begin by equipping our ambient category $\BC$ with a closed monoidal structure 
$(\BC, \Pt, \mon,\monto)$, meant to abstract from the internal mechanisms of variable
management~\cite{Fiore08}. Monoidal closedness yields a natural isomorphism 
$(\argument)^{\flat}\c\BC(X \mon Y, Z)\to \BC(X, Y\monto Z)$, which we will need below.

In this setting, we need the following technical
\begin{definition}[Pointed Strength~\cite{Fiore08,HirschowitzHirschowitzEtAl22}]\label{def:pointed-str}
Let us denote by $j$ the forgetful functor from $V/\BC$ to $\BC$, sending a 
morphism $V\to X$ to $X$. Objects of $V/\BC$ are called \emph{(V-)pointed objects} of $\BC$. 
  
A \emph{($\Pt$-)pointed strength} on an endofunctor $F\c\BC\to\BC$ is a family of morphisms
$\mathrm{st}_{X,Y}\c FX\mon jY\to F(X\mon jY)$, natural in $X\in \BC$ and
$Y\in \Pt /\BC$, such that the following diagrams commute:
\begin{equation*}
\!\!\!\begin{tikzcd}[column sep=0em, row sep=normal]
 & FX\\
 FX\mon\Pt  & & F(X\mon\Pt) 
 \arrow[iso, rotate=180, from=1-2, to=2-1]
 \arrow["\mathrm{st}_{X,\Pt}", from=2-1, to=2-3]
 \arrow[iso, from=1-2, to=2-3] 
\end{tikzcd}
\quad\!\!
\begin{tikzcd}[column sep=3.5em, row sep=normal]
  (F X\mon jY)\mon jZ & F(X\mon jY)\mon jZ & F((X\mon jY)\mon jZ)\\
  F X\mon (jY\mon jZ) & & F(X\mon (jY\mon jZ))
  \arrow[from=1-1,to=2-1, iso] 
  \arrow[from=1-1,to=1-2, "\mathrm{st}_{X,Y}\mon jZ"]
  \arrow[from=1-2,to=1-3, "\mathrm{st}_{X\mon jY,Z}"] 
  \arrow[from=2-1,to=2-3, "\mathrm{st}_{X,Y\mon Z}"]
  \arrow[from=1-3,to=2-3, iso]
\end{tikzcd}
\end{equation*}
(eliding the names of the canonical isomorphisms).
\end{definition}
We then assume that $\Sigma_\val$ has the form $V+\Sigma_\val'$ and that the functor 
$\Sigma_\val'+\Sigma_\com\Delta$ is $V$-strong. This guarantees that 
the initial algebra $\mS$ (if it exists) is a monoid~\cite[Theorem 4.1]{FiorePlotkinEtAl99},
whose multiplication we denote as $\oname{subst}\c\mS\mon\mS\to\mS$.

We modify our framework of separable abstract HO-GSOS laws slightly by replacing
\eqref{eq:separated1} with
\begin{align}\label{eq:separated4}
\rho_{X,Y}^\val\c&\Sigma_\val (jX\times (jX\monto Y)) \to D(jX,\Sigma^\star (jX+Y)),
\end{align}
    dinatural in $X \in \Pt/\BC$ and natural in $Y\in | \BC|$.
    
    We then redefine 
    $\gamma^\val$ as follows: 
\begin{align*}
\mSv
	\xto{\Sigma_\val\brks{\id,\oname{subst}^{\flat}}}
\Sigma_\val(\mS\times(\mS\monto\mS))
	\xto{\rho^\val}
D(\mS,\Sigma^\star(\mS+\mS))
	\xto{D(\id,\nabla^\klstar)} 
D(\mS,\mS)
\end{align*}
and, as before, obtain $\gamma^\com$ from it as the unique such morphism that 
the diagram
\begin{equation*} %
\begin{tikzcd}[column sep=5ex, row sep=4ex]
\Sigma_\com(\mSv+\mSc,\mS)
	\dar["{\Sigma_\com(\brks{\iota,\gamma^\val+\gamma^\com},\id)}"']
	\ar[rr,"{\Sigma_\com(\iota,\id)}"] 
& 
&[2ex]
\mSc
	\dar["\gamma^\com"]	
\\
\Sigma_\com(\mS\times {B(\mS,\mS)},\mS)
  \rar["{\rho^\com }"] 
&
T\Sigmas(\mS+\mS)
	\rar["T\nabla^\klstar"]
&
T\mS
\end{tikzcd}
\end{equation*}
commutes. Using the new operational model $(\gamma^\val,\gamma^\com)$, we obtain the 
multistep semantics $\beta$ using the same formula~\eqref{eq:multi}.

\subsection{Strongly Separable HO-GSOS for the $\lambda$-calculus}
Let us spell out, how the above applies to the (call-by-name) $\lambda$-calculus.

\myparagraph{Category}
Following \mbox{\citet{FiorePlotkinEtAl99}}, let $\fset$ be the category of finite cardinals, the skeleton of the
category of finite sets. The objects of $\fset$ are the sets $n=\{0,\dots,n-1\}$ with $n \in \Nat$, and 
morphisms~${n\to m}$ are functions. Let $\BC$ be the category of presheaves $\Set^\fset$.
Intuitively, the objects of~$\Set^\fset$ are families~$(X(n))_{n\in\Nat}$ of sets
where each $X(n)$ is meant to parametrically depend on $0,\ldots,n-1$ as free variables. 

The process of substituting terms for variables can be treated at the abstract level of presheaves as follows. For every presheaf $Y \in \Set^\fset$, there is a functor $- \mathbin{\mon} Y \c \Set^\fset \to \Set^\fset$ given by
\begin{equation}\label{eq:tensor}
 (X \mathbin{\mon} Y)(m) = \int^{n \in \fset} X(n) \times (Y(m))^{n} = \Bigl(\coprod_{n\in \fset} X(n)\times (Y(m))^n\Bigr)\Big/_{\displaystyle\approx}
\end{equation}
where $\approx$ is the equivalence relation generated by all pairs
\[
  (x, y_0,\ldots, y_{n-1}) \approx (x', y_0',\ldots, y_{k-1}')
\]
such that $(x, y_0,\ldots, y_{n-1})\in X(n)\times (Y(m))^n$, $(x', y_0',\ldots, y_{k-1}')\in X(k)\times Y(m)^k$ and there
exists $r\c n\to k$ satisfying $x'=X(r)(x)$ and $y_{i}=y'_{r(i)}$ for
$i=0,\ldots, n-1$.  An equivalence class in~\eqref{eq:tensor}
can be thought of as a term $x\in X(n)$ with $n$ free variables,
together with~$n$ terms $y_0,\ldots,y_{n-1}\in Y(m)$ to be substituted
for them. The equivalence relation then captures the idea that the outcome of the substitution should be invariant under renamings that reflect equalities among $y_0,\ldots,y_{n-1}$; for instance, if $y_i=y_j$ and $r\c n\to n$ is the bijective renaming that swaps $i$ and $j$, then substituting $y_0,\ldots, y_{n-1}$ for $0,\ldots,n-1$ in the term $X(r)(x)$ should produce the same outcome. Varying~$Y$,
one obtains the \emph{substitution tensor}
\[\argument \mon \argument \c \Set^\fset \times \Set^\fset \to \Set^\fset,\] 
which makes $\Set^\fset$ into a (non-symmetric) monoidal category with the following 
\emph{presheaf of variables}~$V$ as the unit object: $V(n) = \{0,\ldots,n-1\}$. 

Monoidal closedness of $\Set^\fset$ is witnessed by the fact that for every $X\in |\Set^\fset|$ the functor
$- \mathbin{\mon} X \c \Set^\fset \to \Set^\fset$ has a right adjoint given by
\begin{equation*}
  X\monto \argument \c \Set^\fset \to \Set^\fset, \qquad (X\monto Y)(n) = \int_{m \in \fset} [(X(m))^{n}, Y(m)] = \oname{Nat}(X^{n} , Y).
\end{equation*}
An element of $(X\monto Y)(n)$, viz.\ a natural in $m\in |\fset|$ family
of maps $(X(m))^n\to Y(m)$, is thought of as describing the substitution of $X$-terms in $m$ variables for the $n$ variables of some fixed ambient term, resulting in  a $Y$-term in $m$ variables.

\myparagraph{Syntax}
The following \emph{context extension} endofunctor $\delta \c \Set^\fset \to\Set^\fset$
is used for including the lambda abstraction to the signature. On objects: $\delta X (n) = X(n + 1)$ and $\delta X (h) = X(h + \id_{1})$, 
on morphisms: $h \c X \to Y$, $(\delta h)_{n} = (h_{n + 1}\c X(n + 1)\to  Y(n + 1))$.
Informally, the elements of $\delta X(n)$ are terms arising by binding the last variable is a term with $n+1$ free variables.

We define the value and the computation signatures as follows:
\begin{align*}
\Sigma_\val X = V \times{\textstyle\coprod}_{k\in\Nat} X^k +\delta X && \Sigma_\com(X,Y) = X \times Y  
\end{align*}
As in the case of extended combinatory logic, $\Sigma_\com$ covers application, 
while $\Sigma_\val$ covers two types of values: expressions of the form $(\ldots(x t_1)\ldots) t_k$ which we write as $x*(t_1,\ldots,t_k)$
  with $x$ being a variable, and $\lambda$-abstractions.
The initial object $\mu\Sigma$ is known to model the terms of $\lambda$-calculus
(with free variables)~\cite{FiorePlotkinEtAl99}, although our case is slightly 
different in that we distinguish generic applications $st$ from those that 
are representable as $x*(t_1,\ldots,t_k)$. Modulo that, $(\mS)(n)$ are the 
$\lambda$-terms over $0,\ldots,n-1$ as free variables, under $\alpha$-equivalence.
Our requirement of $V$-pointed strength for $(V \times\coprod_{k\in\Nat\smin\{0\}} X^k+\delta X) + X\times X$ means essentially 
that substitution is definable by structural 
recursion, and the result that $\mS$ is a monoid, means that $\lambda$-terms are 
closed under substitution. In terms of the presentation~\eqref{eq:tensor}, for every 
$(t,s_0,\ldots,s_{n-1})\in \mS(n)\times (\mS(m))^n$, $\oname{subst}([(t,s_0,\ldots,s_{n-1})]_{\approx})$
coherently computes the substitution $t[s_0/0,\ldots,s_{n-1}/n-1]$.

\myparagraph{Semantics} We define the behaviour functor $D$ via $D(X,Y) = V \times\coprod_{k\in\Nat} Y^k+Y^X$. Let us spell out~\eqref{eq:separated4} and \eqref{eq:separated2}.
The components of the first transformation are obtained by cotupling 
\begin{equation*}
\begin{tikzcd}[column sep=normal, row sep=normal]
V \times\coprod_{k\in\Nat} (jX\times (jX\monto Y))^k 
\dar{\inl\comp(\id\times \coprod_{k\in\Nat} (\eta\comp\inl\comp\fst)^k)}\\
V \times\coprod_{k\in\Nat} (\Sigma^\star (jX+Y))^k+(\Sigma^\star (jX+Y))^{jX}              
\end{tikzcd}
\end{equation*}
and
\begin{equation*}
\begin{tikzcd}[column sep=20ex,row sep=5ex]
\delta(jX\times (jX\monto Y))
	\rar{\delta\snd}
&
\delta(jX\monto Y)
	\ar[ld,"{\curry(\kappa_{X,Y})}"',pos=.8,bend right=7,start anchor={[yshift=-3pt]west}]
\\ Y^{jX}
	\rar["{\inr\comp(\eta\comp\inr)^{jX}}"]
&
V \times\coprod_{k\in\Nat} (\Sigma^\star (jX+Y))^k + (\Sigma^\star (jX+Y))^{jX}
\end{tikzcd}
\end{equation*}
where $(\kappa_{jX,Y})_n\c\oname{Nat}(jX^{n+1},Y)\times jX(n)\to Y(n)$ sends $\alpha\in\oname{Nat}(jX^{n+1},Y)$ and $t\in jX(n)$ 
to $\alpha_n(v_0,\ldots,v_{n-1},t)$ and $v_i = X_n(i)$ (recall that $X\in V/\BC$, i.e.\ 
is a natural transformation from $V$ to~$jX$). Thus, $\rho^\val$ encodes 
small-step rules:
\begin{gather*}
\frac{}{x*(t_1,\ldots,t_k)\dar x*(t_1,\ldots,t_k)}\quad (x\in V(n), t_1,\ldots,t_k\in X(n))\\[1.5ex]
\frac{t[s/n]=t'}{\lambda n.\,t\xto{s} t'}\quad (t\in jX(n), s\in X(n-1))
\end{gather*}
The notation $\dar$ refers to the left summand of $D$, and the first rule
associates this behaviour with variables and variable applications. The notation $\xto{s} t$ refers to the right 
summand of $D$, and the second rule describes the behaviour of $\lambda$-abstraction.
Unlike the rule for values in \xCL (\autoref{fig:rules}), this rule requires 
a premise, which motivates the inclusion of $jX\monto Y$ in \eqref{eq:separated4}.
This premise provides the information of how substitutions act on $t$, specifically
how the substitution $t[s/n]$ is computed. The latter is, in fact, an abbreviation 
for $t[0/0,\ldots,n-1/n-1,s/n]$, which makes it clear why $X$ must be an object of $V/\BC$ 
and not of $\BC$. Indeed, when applying $[i/i]$ to $t$, only the right $i$ is defined for 
$t\in X(n)$ with $X\in|\BC|$, but the left one requires a transformation $V\to X$
reifying variables into~$X$.

The transformation \eqref{eq:separated2} instantiates as
$\rho_{X,Y}^\com\c(X\times ((V \times{\textstyle\coprod}_{k\in\Nat} Y^k+Y^X)+Y))\times X\to\Sigma^\star (X+Y)$,
(assuming that the monad $\BBT$ is identity). Now that binding operators 
are involved, the intended transformation is the one that captures the following rules in a 
straightforward manner:
\begin{align*}
\frac{t\dar x*(t_1,\ldots,t_k)}{ts\dar x*(t_1,\ldots,t_k,t)}\qquad\quad
\frac{t\xto{s} t'}{ts\to t'}\qquad\quad
\frac{t\to t'}{ts\to t's}
\end{align*}
By \autoref{rem:total-sep}, we obtain a strongly separated abstract HO-GSOS law.
The simplest way to extend it to a law w.r.t.\ an $\omega$-continuous monad $\BBT$
is to take $T$ to be the pointwise powerset functor $\PSet_\star$, i.e.\ 
$\PSet_\star(X) = \PSet\comp X$, and apply \autoref{pro:transfer}. The powerset 
monad is an $\omega$-continuous monad on~$\Set$, with a standard $\omega$-continuous 
distributive law over $(\argument)\times Y$. Since all relevant constructions 
are pointwise, we obtain that $\PSet_\star$ is $\omega$-continuous, we obtain 
an $\omega$-continuous distributive law~\eqref{eq:separated3}, and also inherit 
the strong separation condition per \autoref{pro:transfer}: we have it for $\PSet$
and $(X,Y)\mto X\times Y$ in $\Set$ by \autoref{ref:transfer}, and extend it to 
$\PSet_\star$ and $\Sigma_\com$ pointwise.

\subsection[Big-Step SOS and the λ-calculus]{Big-Step SOS and the $\lambda$-calculus}
The modification of the big-step SOS law~\eqref{eq:big-step} for our present 
setup is:
\begin{align*}
	\xi\c\Sigma_\com(\Sigma_\val (X\times(\Sigmas X\monto Y)), X)\to T\Sigma^\star (X + Y).
\end{align*}
The operational model~\eqref{eq:bss} is computed using the following 
modification of~\eqref{eq:bs-xi}:
\begin{align*}
	\hat\zeta = \mu f.\, [\eta,f^\klstar\comp T\nabla^\klstar\comp\xi^\klstar\comp T\Sigma_\com(\Sigma_\val\brks{\id,(\oname{subst}\comp(\id\times\mu))^{\flat}},\id) \comp\dl\comp\Sigma_\com(f,\id)]\comp\iota^\mone.
\end{align*} 
The key change is the inclusion of the morphism $\brks{\id,(\oname{subst}\comp(\id\times\mu))^{\flat}}\c\mS\to\mS\times(\Sigmas\mS\monto\mS)$,
which creates the new expected part of the input for $\xi$ -- informally, given 
a term $t$ from $\mS(n)$, we render $\Sigmas\mS\monto\mS$ as the space of 
substitution actions $t[-/0,\ldots,-/n-1]$, awaiting terms from~$\Sigmas\mS$.

The updated translation~\eqref{eq:xi-constr} of a separated abstract
HO-GSOS law to a big-step SOS law is as follows:
\begin{equation}\label{eq:xi-constr-lam}
\!\!\!\begin{tikzcd}[column sep=12ex,row sep=5ex]
\Sigma_\com(\Sigma_\val (X\times(\Sigmas X\monto Y)), X)
	\rar{\Sigma_\com(\Sigma_\val(\eta\times\id),\eta)}
&
\Sigma_\com(\Sigma_\val(\Sigmas X\times(\Sigmas X\monto Y)),\Sigmas X)
	\ar[ld,"{\Sigma_\com(\brks{\iota^\val\comp\Sigma_\val\fst,\, \rho^\val},\id)}"',pos=.9,bend right=7,start anchor={[yshift=-3pt]west}]
\\ \Sigma_\com(\Sigmas X\times D(\Sigma^\star X,\Sigma^\star (\Sigmas X + Y)),\Sigmas X)
	\rar["{T[\Sigmas\inl,[\Sigmas\inl,\eta\comp\inr]^\klstar]^\klstar\comp\rho^{\com\val}}"]
&
T\Sigma^\star(X+Y)
\end{tikzcd}
\end{equation}
Here, we treat $\Sigmas X$ as an object of $V/\BC$, which is justified since, by assumption, $V$ is a coproduct summand of $\Sigmas X$. This allows us to invoke $\rho^\val$.

Again, if $\xi$ is obtained from $\rho^\val$ and $\rho^\com$ by~\eqref{eq:xi-constr-lam},
the above definition of~$\hat\zeta$ via $\xi$ reduces to a definition via $\rho^\val$
and $\rho^\com$. The following is the 
analogue of \autoref{lem:bs-fix}.
\begin{proposition}\label{lem:bs-fix-lam}
Let $\rho^\val$ and $\rho^\com,\dl$ be as in \autoref{sec:bind-gsos}, 
and let $\dl$ be an $\omega$-continuous distributive law. Let~$\xi$ be defined by~\eqref{eq:xi-constr-lam}. Then $\hat\zeta$ satisfies the 
equation~\eqref{eq:bs-xi-fix}.
\end{proposition}
Our main result (\autoref{the:main}) can now be reestablished (the proof relies 
on \autoref{lem:bs-fix-lam}, but otherwise remains essentially unchanged, because 
it does not depend on how $\rho^\val$ is defined).
\begin{theorem}\label{the:main-lam}
Let $\rho^\val$ and $\rho^\com,\dl$ be as in \autoref{sec:bind-gsos}, 
and let $\dl$ be an $\omega$-continuous distributive law. Let~$\xi$
be defined by~\eqref{eq:xi-constr-lam}. 
Then $\hat\zeta = \hat\beta$.
\end{theorem}
Applying these results to the $\lambda$-calculus example, we obtain the equivalence 
of small-step semantics and the following big-step semantics:
\begin{gather*}%
\frac{}{\lambda n.\,t\Dar \lambda n.\,t}\quad (t\in X(n)) \qquad
\frac{}{x*(t_1,\ldots,t_k)\Dar x*(t_1,\ldots,t_k)}\quad (x\in V(n), t_1,\ldots,t_k\in X(n))\\[2ex]
\frac{t\Dar x*(t_1,\ldots,t_k)\qquad x*(t_1,\ldots,t_k,s)\Dar v}{ts\Dar v}\quad (x\in V(n), t_1,\ldots,t_k\in X(n))\\[2ex]
\frac{t\Dar\lambda n.\,t'\quad t'[s/n]\Dar v}{ts\Dar v} \quad (t,s\in X(n-1),t'\in X(n))
\end{gather*}%

\section{Conclusions and Further Work}

Building on recent advances in higher-order mathematical operational semantics, our present work addresses the well-known question of equivalence between small-step and big-step operational semantics. This equivalence arises in various settings, serving both as a sanity check and an essential tool for program analysis. Rather than developing syntax-driven recipes for rule transformations, our approach is rooted in semantic ideas going back to \citet{TuriPlotkin97}. Specifically, we represent the syntax and behaviour of programs as functors, with operational semantics rules modelled as (di-)natural transformations. This abstraction allowed us to systematically define small-step and big-step semantics and formulate a key condition, which we call \emph{strong separation}, enabling us to prove the desired equivalence. We provided numerous examples demonstrating that our framework accommodates a wide range of features specified through operational semantics.  

Our refinement of the general notion of abstract HO-GSOS is motivated by the goals outlined above. While other formats -- including alternative refinements of abstract HO-GSOS -- may also be viable, we believe our approach achieves a balanced trade-off between expressiveness and practical applicability. This is substantiated by the following points:
\begin{itemize}
\item Our translation produces big-step rules, which are arguably in accord with the common understanding of ``big-step'', as illustrated by case studies (for contrast, see Bernstein~\cite{Bernstein98}, who, motivated by the problem of congruence of program equivalence, abstracted big-step semantics in a way that departs from the established format of big-step rules and judgments).
\item Our framework is grounded in several structural assumptions -- monads, enrichment in $\omega$-cpos, and strong separability -- all of which are explicitly justified. For instance, enrichment is required to interpret recursive definitions, and strong separability is necessary to ensure equivalence properties.
\item While more general rule formats than strongly separated HO-GSOS are conceivable, they often involve intricate constructions and non-elementary conditions (see e.g.~Assumptions III.1–III.4 in~\cite{Ciobaca13}). Our design reflects a deliberate choice to favor conceptual simplicity and usability without compromising the applicability to natural examples.
\end{itemize}
Moving forward, we plan to extend our framework to cover other flavours of semantics, particularly stateful and quantitative, such as probabilistic semantics. With our approach, we hope to gain insights into challenging cases, such as McCarthy's $\amb$ operator~\cite{McCarthy59}, helping one to better understand its sophisticated behaviour. In particular, since the relevant small-step semantics is essentially quantitative (reductions are indexed by numbers to ensure fairness) this raises hopes that a carefully chosen monad could effectively capture this behaviour. Additionally, we plan to complement our current Haskell translations and examples with an implementation in Agda, accommodating not only constructions, but also formal proofs. Another interesting direction for future work is the abstract reverse translation of big-step semantics into small-step semantics, a topic already explored in the literature~\cite{AmbalSchmittEtAl20,GallagherHermenegildoEtAl20}.

\begin{acks}
The authors acknowledge funding by the Deutsche Forschungsgemeinschaft (DFG, German Research Foundation) -- project numbers 527481841. 
\end{acks}

\clearpage

\bibliographystyle{ACM-Reference-Format}
\bibliography{references}

\clearpage
\appendix

\section{Omitted Proofs}
\subsection{Proof of \autoref{pro:theta-om}}
We have the following $\Sigma$-algebra structure on $\mS \times B(\mS,\mS)$:
\begin{equation*}
\Sigma(\mS \times B(\mS,\mS))
	\xto{\brks{\iota \comp \Sigma\fst, B(\id,\nabla^\klstar)\comp\rho'}} \mS \times B(\mS,\mS).              
\end{equation*}
This induces a universal $\Sigma$-algebra morphism $\mS \to \mS \times B(\mS,\mS)$, necessarily of the form $\brks{\id, \gamma}$~\cite{GoncharovMiliusEtAl23}, where $\gamma$
is the operational model for $\rho'$, i.e.\ $\gamma$ is the unique such morphism that the following diagram commutes:
\begin{equation}\label{eq:gamma-unique}
	\begin{tikzcd}[column sep=1.5em]
		\Sigma(\mS) 
		\arrow[d, "{\Sigma\brks{\id,\gamma}}"'] 
		\arrow[rr, "\iota"] 
		&&[3em] 
		\mS 
		\arrow[d, "\gamma"] 
		\\
		{\Sigma(\mS\times {B(\mS,\mS)})} 
		\arrow[r, "\rho'"]
		& 
		{B(\mS,\Sigma^\star (\mS+\mS))} 
		\arrow[r, "{B(\id,\nabla^\klstar)}"] 
		& 
		{B(\mS,\mS)}                
	\end{tikzcd}
\end{equation}
We will show that the last diagram can be transformed equivalently, so that the result coincides with~\eqref{eq:op-mod} -- this will immediately imply the claim.
	Let us first show that 
	$\rho'_{\mS,\mS}=\rho_{\mS,\mS} \comp \Sigma'(\id,\fst)$
	as follows: %
	\begin{flalign*}
		&&\rho' =&\;
		B(\id, \Sigma^\star [\inl, \id]) \comp \rho_{\mS, \mS + \mS} \comp \Sigma'(\id\times B(\id,\inr), \inl \comp \fst)&\\
		&&\;=&\;B(\id, \Sigma^\star [\inl, \id])\comp B(\id, \Sigma^\star\inr)\comp B(\id, \Sigma^\star\nabla)\comp \rho_{\mS, \mS + \mS}\\
		&&&\qquad  \comp \Sigma'(\id\times B(\id,\inr), \inl \comp \fst)&\\
		&&\;=&\;B(\id, \Sigma^\star [\inl,\id]) \comp B(\id, \Sigma^\star\inr) \comp \rho_{\mS, \mS}\\
		&&&\qquad \comp \Sigma'(\id \times B(\id, \nabla), \nabla)\comp \Sigma'(\id\times B(\id,\inr), \inl \comp \fst)&\by{naturality~$\rho$}\\
		&&\;=&\;\rho_{\mS, \mS} \comp \Sigma'(\id, \fst).
	\end{flalign*}
	Now, using the calculation
	\begin{flalign*}
		&&B(\id,\nabla^\klstar) \comp \rho' \comp \Sigma \brks{\id, \gamma} =&\;
		B(\id,\nabla^\klstar) \comp \rho_{\mS, \mS} \comp \Sigma'(\id, \fst) \comp \Sigma \brks{\id, \gamma}&\\
		&&\;=&\;B(\id,\nabla^\klstar) \comp \rho_{\mS, \mS} \comp \Sigma'(\id, \fst) \comp \Sigma' (\brks{\id, \gamma},\brks{\id, \gamma})\\
		&&\;=&\;B(\id,\nabla^\klstar) \comp \rho_{\mS, \mS} \comp \Sigma' (\brks{\id, \gamma},\id)
	\end{flalign*}
	the diagram~\eqref{eq:gamma-unique} can be transforms into~\eqref{eq:ho-var-conv}, as desired.
\qed

\subsection{Proof of \autoref{pro:gamma_c}}
Equation \eqref{eq:gammas} gives an idea, how to define $\gamma^\com$ using 
\eqref{eq:op-mod-c}. Let $\gamma^\com$ be the composition
\begin{align*}
\mSc
	\xto{\Sigma_\com(\brks{\id,\gamma},\id)}
\Sigma_\com(\mS\times {B(\mS,\mS)},\mS)
  	\xto{\rho^\com} 
T\Sigmas(\mS+\mS)
  	\xto{T\nabla^\klstar}
T\mS.
\end{align*}
We then can prove~\eqref{eq:gammas}:
\begin{flalign*}
&& \gamma^\val+\gamma^\com &\;= [\inl\comp D(\id,\mu) \comp \rho^\val, \inr\comp T\nabla^\klstar\comp \rho^\com\comp \Sigma_\com(\brks{\id,\gamma},\id)]&\\*
&&  &\;= [\inl\comp D(\id,\mu)\comp \rho^\val, B(\id,\nabla^\klstar)\comp\inr\comp \rho^\com\comp \Sigma_\com(\brks{\id,\gamma},\id)]\\
&&  &\;= B(\id,\nabla^\klstar)\comp [\inl\comp D(\id,\Sigmas\inl)\comp \rho^\val, \inr\comp \rho^\com\comp \Sigma_\com(\brks{\id,\gamma},\id)]\\
&&  &\;= B(\id,\nabla^\klstar)\comp (D(\id,\Sigmas\inl)\comp  \rho^\val+ \rho^\com)\comp \Sigma'(\brks{\id,\gamma},\id)\\
&&  &\;= B(\id,\nabla^\klstar)\comp \rho\comp \Sigma'(\brks{\id,\gamma},\id)\\
&&  &\;= \gamma\comp\iota.
\end{flalign*}
By applying \eqref{eq:gammas} to~\eqref{eq:op-mod-c}, we obtain a diagram 
that commutes by the definition of~$\gamma^\com$. We are left to show that $\gamma^\com$
is the only morphism, for which~\eqref{eq:op-mod-c} commutes. Let $g$ be a
morphism replacing $\gamma^\com$ for which~\eqref{eq:op-mod-c} commutes. It then 
follows that the diagram
\begin{equation*}
\begin{tikzcd}[column sep=6ex, row sep=4ex]
\Sigma'(\mS,\mS)
	\dar["{\Sigma'(\iota^\mone,\id)}"']
	\ar[rr,"\iota"] 
& 
&[1ex]
\mS
	\dar["\iota^\mone"]	
\\
\Sigma'(\mSv+\mSc,\mS)
	\dar["{\Sigma'(\brks{\iota,\gamma^\val+ g},\id)}"']
	\ar[rr,"{\Sigma'(\iota,\id)}"] 
& 
&[1ex]
\Sigma'(\mS,\mS)
	\dar["\gamma^\val+ g"]	
\\
\Sigma'(\mS\times {B(\mS,\mS)},\mS)
  \rar["{\rho}"] 
&
B(\mS,\Sigma^\star (\mS+\mS))
	\rar["{B(\id,\nabla^\klstar)}"]
&
B(\mS,\mS)
\end{tikzcd}
\end{equation*}
commutes as well. Since $\Sigma'(\brks{\iota,\gamma^\val+ g},\id)\comp\Sigma'(\iota^\mone,\id) = \Sigma'(\brks{\id,(\gamma^\val+ g)\comp\iota^\mone},\id)$,
we conclude that $(\gamma^\val+ g)\comp\iota^\mone$ satisfies the characteristic property~\eqref{eq:op-mod}
of $\gamma$. Therefore, $\gamma = (\gamma^\val+ g)\comp\iota^\mone$. By combining it 
with~\eqref{eq:gammas}, we obtain $\gamma^\val+ g = \gamma^\val+\gamma^\com$,
and therefore $g = T\nabla^\klstar\comp\rho^\com\comp\Sigma_\com(\brks{\iota,\gamma^\val+\gamma^\com},\id)\comp\Sigma_\com(\iota^\mone,\id)$,
using the assumption about~$g$. We obtain that any $g$, for which \eqref{eq:op-mod-c} commutes 
equals to an expression that does not depend on $g$. Hence, there is at most one such $g$.
\qed

\subsection{Proof of \autoref{pro:transfer}}

The strong separation condition~\eqref{eq:str_sep} follows from the assumptions,
as the following diagram depicts: 
{%
\begin{equation*}
\begin{tikzcd}[column sep=-6em, row sep=0ex]
&[1em] 
&[-1em]
\Theta(X\times D(X,Y),X\times TY,X)
	\ar[dll,"\theta"', bend right=2]
	\ar[drr,"\theta",  bend left=2]
	\ar[dd,"{\Theta(\eta,\tau,\id)}"]
&[-1em]
&[2.5em]
\\[6ex]
\makebox[6em][l]{$\mathrlap{\Sigma_\com(X\times D(X,Y) + X\times TY,X)}$}
	\ar[ddd,"\iso"']
	\ar[ddr,"{\Sigma_\com(\eta+\tau,\id)}", pos=.4]
&[0em] 
&[-1em]
&[-1em]
&[1.35em]
\makebox[6em][r]{$\mathllap{\Sigma_\com(X\times D(X,Y) + X\times TY,X)}$}
	\ar[ddl,"{\Sigma_\com(\eta+\tau,\id)}"', pos=.4]
	\ar[dddddd,"{\Sigma_\com(\eta\comp\fst\,\klplus\snd,\inl)}"']
\\[1ex]
&[5em] 
&[-1em]
\Theta(T(X\times D(X,Y)),T(X\times Y),X)
	\ar[ddddd]
	\ar[dl,"\theta"']
	\ar[dr,"\theta"]
&[-1em]
&[2em]
\\[3ex]
& 
\Sigma_\com(T(X\times D(X,Y)) + T(X\times Y),X)
	\ar[ddd,"{\Sigma_\com(\id\klplus\id,\id)}"]
&
&
\Sigma_\com(T(X\times D(X,Y)) + T(X\times Y),X)
	\ar[ddd,"{\Sigma_\com(\id\klplus\id,\id)}"']
\\[1ex]
\Sigma_\com(X\times B(X,Y),X)
	\ar[d,"{\Sigma_\com(\id\times(\id\klplus\id),\id)}"] 
\\[6ex]
\makebox[10em][l]{$\mathrlap{\Sigma_\com(X\times T(D(X,Y) + Y),X)}$}
	\ar[dd,"{\Sigma_\com(\tau,\id)}"'] 
\\[1ex]
&
\Sigma_\com(T(X\times D(X,Y) + X\times Y),X)
	\ar[dd,"\dl"]
	\ar[dl,"\iso"'] %
&
&
\Sigma_\com(T(X\times D(X,Y) + X\times Y),X)
	\ar[dd,"\dl"']
	\ar[dr,"{\!\!\Sigma_\com(T(\fst+\snd),\inl)}",pos=.45]
\\[6ex]
\makebox[10em][l]{$\mathrlap{\Sigma_\com(T(X\times (D(X,Y) + Y)),X)}$}
	\ar[ddr,"\dl"',bend right=30]
& &T\Theta(X\times D(X,Y),X\times Y,X)
	\ar[dl,"T\theta"']
	\ar[dr,"T\theta"]& &
\makebox[6em][r]{$\mathllap{\Sigma_\com(T(X+Y),X+Y)}$}
	\ar[ddl,"{\dl}",bend left=30]
\\[4ex]
& T\Sigma_\com(X\times D(X,Y) + X\times Y,X)
	\ar[d,"\iso"] %
& &
T\Sigma_\com(X\times D(X,Y) + X\times Y,X)
	\dar["{T\Sigma_\com(\fst+\snd,\inl)}"']
\\[5ex]
& T\Sigma_\com(X\times (D(X,Y) + Y),X)
	\ar[dr,"T\rho^\com"']
& &
T\Sigma_\com(X+Y,X+Y)
	\ar[dl,"{T\iota^\com\comp T\Sigma_\com(\eta,\eta)}"]
&
\\[2ex]
& & T\Sigma^\star (X+Y) &
\end{tikzcd}
\end{equation*}
}

\qed

\subsection{Proof of \autoref{lem:xi-mS}}
For the sake of succinctness, we abbreviate $\Sigmas(\mS)$ as $\mu^2\Sigma$.
The claim follows from commutativity of the diagram:
\begin{equation*}
\begin{tikzcd}[column sep=-4em, row sep=0ex]
\Sigma_\com(\mSv,\mS)
	\dar["{\Sigma_\com(\Sigma_\val\eta,\eta)}"']
	\ar[rr,"{\Sigma_\com(\brks{\iota^\val,\, \rho^\val},\id)}"]
&[3.5em] &[1.5em]
\Sigma_\com(\mS\times D(\mS,\mu^2\Sigma),\mS)
	\ar[dddddd, "{\Sigma_\com(\id\times D(\id,\mmul),\id)}"]
\\[4ex]
\Sigma_\com(\Sigma_\val(\mu^2\Sigma),\mu^2\Sigma)
	\ar[dr,"{\Sigma_\com(\brks{\iota^\val,\Sigma_\val\mmul},\id)}", bend right=10, pos=.3]
    \ar[ddd,"{\Sigma_\com(\brks{\iota^\val,\,\rho^\val},\id)}"']
& &[3.25em]
\\[2ex]
&
\Sigma_\com(\mu^2\Sigma\times\Sigma_\val\mS,\mu^2\Sigma)\qquad\quad
	\ar[d, "{\Sigma_\com(\id\times \rho^\val,\id)}"']
&
\\[5ex]
&
\Sigma_\com(\mu^2\Sigma\times D(\mS,\Sigmas(\mu^2\Sigma),\mu^2\Sigma)\quad
	\ar[dd,"{\Sigma_\com(\id\times D(\mmul,\id),\id)}",pos=.3]
	\ar[ruuu, "{\Sigma_\com(\mmul\times\id,\mmul)}", bend right=25, pos=.7]
&
\\[0ex]
\Sigma_\com(\mu^2\Sigma\times D(\mu^2\Sigma,\Sigmas(\mu^2\Sigma)),\mu^2\Sigma)
	\ar[dd,"{\Sigma_\com(\id\times D(\id, \mmul),\id)}"']
	\ar[dr, "{\Sigma_\com(\id\times D(\id,\Sigmas\mmul),\id)}", bend right=20,pos=.3]
& 
&
\\[8ex]
& 
\Sigma_\com(\mu^2\Sigma\times D(\mu^2\Sigma,\mu^2\Sigma),\mu^2\Sigma)
	\ar[dd, "{\Sigma_\com(\id\times D(\id,\mmul),\id)}",pos=.15]
&
\\[4ex]
\Sigma_\com(\mu^2\Sigma\times D(\mu^2\Sigma,\mu^2\Sigma),\mu^2\Sigma)
	\ar[dd,"{\rho^{\com\val}}"']
	\ar[dr, "{\Sigma_\com(\id\times D(\id,\mmul),\id)}", bend right=20,pos=.3]
& &
\Sigma_\com(\mS\times D(\mS,\mu\Sigma),\mS)
	\ar[dd,"{\rho^{\com\val}}"]
\\[8ex]
& 
\Sigma_\com(\mu^2\Sigma\times D(\mu^2\Sigma,\mS),\mu^2\Sigma)\quad
	\ar[dd,"\rho^{\com\val}",pos=.3]
&
\\[1ex]
T\Sigmas(\mu^2\Sigma+\mu^2\Sigma)
	\ar[dd,"{T\nabla^{\sharp}}"']
	\ar[dr,"T\Sigmas(\id+\mmul)", bend right=15]
& 
&
T\Sigmas(\mS+\mS)
	\ar[dd,"T\nabla^\klstar"]
\\[5ex]
&
T\Sigmas(\mu^2\Sigma+\mS) 
	\ar[ur,"T\Sigmas(\mmul+\id)", bend right=15]&
&
\\[2ex]
T(\mu^2\Sigma)
	\ar[rr,"T\mmul"]
& &
T(\mS)
\end{tikzcd}
\end{equation*}
The top cell commutes, because $\iota^\val$ satisfies the equation $\iota^\val\comp\Sigma_\val\mmul = \mmul\comp\iota^\val$,
by naturality of $\Sigma_\val$ and because $\mmul\comp\eta=\id$. The bottom cell 
commutes because $\mmul\comp \nabla^\klstar = \mmul\comp \mmul\comp\Sigmas\nabla = \mmul\comp\Sigmas\mmul\comp\Sigmas\nabla = \mmul\comp \Sigmas\nabla\comp \Sigmas(\mmul+\mmul) = \nabla^\klstar\comp\Sigmas(\mmul+\mmul)$.
The remaining three cells on the left-hand side of the diagram commute by 
dinaturality of $\rho^\val$, by definition of Kleisli lifting for $\Sigmas$ and 
by naturality of~$\rho^{\com\val}$ correspondingly. The remaining large cell on 
the right-hand side commutes by naturality and dinaturality of~$\rho^{\com\val}$.  
\qed

\subsection{Proof of \autoref{lem:beta-delta}}
As a preliminary step, we show that the following diagram commutes:
\begin{equation}\label{eq:beta-delta}
\begin{tikzcd}[column sep=normal, row sep=normal]
\mSc
	\rar["\beta"]
	\dar["{\Sigma_\com(\hat\gamma^\com,\id)}"'] 
&
T(\mSv)
\\               
\Sigma_\com(T\mS,\mS)
	\rar["\dl"] 
&
T(\mSc)
	\uar["\beta^\klstar"']
\end{tikzcd}
\end{equation}
To that end we use the fact that, by \autoref{ass:compl}, $\mSc$ is a coproduct of $\Sigma_\com(\mSv,\mS)$ and $\Theta(\mSv,\mSc,\mS)$.
The equation encoded by~\eqref{eq:beta-delta} thus falls into two equations:
\begin{align*}
\beta^\klstar\comp\dl\comp\Sigma_\com(\hat\gamma^\com,\id)\comp\Sigma_\com(\iota^\val,\id) =&\; \beta\comp\Sigma_\com(\iota^\val,\id),\\
\beta^\klstar\comp\dl\comp\Sigma_\com(\hat\gamma^\com,\id)\comp\Sigma_\com(\iota,\id)\comp\theta =&\; \beta\comp\Sigma_\com(\iota,\id)\comp\theta.
\end{align*}
The first one is easy to obtain by unfolding definitions and using the fact that $\dl$ is a distributive law:
\begin{flalign*}
&& \beta^\klstar\comp\dl\comp\Sigma_\com(\hat\gamma^\com,\id)\comp\Sigma_\com(\iota^\val,\id)  &\;= \beta^\klstar\comp\dl\comp\Sigma_\com(\eta\comp\iota^\val,\id) %
= \beta\comp\Sigma_\com(\iota^\val,\id).&&
\end{flalign*}
Let us proceed with the second equation. Using the strong separation condition~\eqref{eq:str_sep} we obtain that the following 
diagram commutes:
\begin{equation*}
\!\!\begin{tikzcd}[column sep=-7.5em, row sep=1ex]
	{\Theta(\mSv,\mSc,\mS)} &[4.5em]&[0em]&[1.75em] {\Sigma_\com(\mSv+\mSc,\mS)} \\[4ex]
	&& {\Theta(T\mS,T\mS,\mS)} \\
	& {\Theta(\mS\times D(\mS,\mS),\mS\times T\mS,\mS)} && {} \\[1ex]
	{\Sigma_\com(\mSv+\mSc,\mS)} & {} && {\Sigma_\com(T\mS+T\mS,\mS)} \\[1ex]
	& & {\Sigma_\com(\mS\times D(\mS,\mS)+\mS\times T\mS,\mS)} & \\[4ex]
	{} & {} &&  \\
	&  && {\Sigma_\com(T(\mS+\mS),\mS+\mS)\hspace{1.5ex}} \\[2ex]
	{} & {\Sigma_\com(\mS\times D(\mS,\mS)+\mS\times T\mS,\mS)} && {T\Sigma_\com(\mS+\mS,\mS+\mS)} \\[4ex]
	{\Sigma_\com(\mS\times B(\mS,\mS),\mS)} &&& {T\Sigma^\star (\mS+\mS)} 
	\arrow["\theta", from=1-1, to=1-4]
	\arrow["{\Theta(\eta\comp\iota^\val,\gamma^\com,\id)}"{pos=.8}, from=1-1, to=2-3,bend left=4]
	\arrow["{\Theta(\brks{\iota^\val,\gamma^\val},\brks{\iota^\com,\gamma^\com},\id)}"'{pos=.7}, from=1-1, to=3-2,bend left=10]
	\arrow["\theta"', from=1-1, to=4-1]
	\arrow["{\Sigma_\com(\eta\comp\iota^\val+\gamma^\com,\id)}", from=1-4, to=4-4]
	\arrow["\theta",bend left=15, from=2-3, to=4-4]
	\arrow["{\Theta(\eta\comp\fst,\snd,\id)}"{pos=0.75}, from=3-2, to=2-3]
	\arrow["{\theta}"{pos=0.55}, from=3-2,to=5-3,bend right=10]
    \arrow["{\Sigma_\com(\eta\comp\fst\,\klplus \snd,\inl)}"'{pos=0.55}, from=5-3,to=7-4,bend left=15]
	\arrow["\theta"', from=3-2, to=8-2]
	\arrow["{\Sigma_\com(\brks{\iota^\val,\gamma^\val}+\brks{\iota^\com,\gamma^\com},\id)}"'{pos=.85}, from=4-1, to=8-2,bend left=20]
	\arrow["{\Sigma_\com(\brks{\iota,\gamma^\val+\gamma^\com},\id)}"'{pos=.35}, from=4-1, to=9-1]
	\arrow["{\Sigma_\com(\id \klplus \id,\inl)}", from=4-4, to=7-4]
  \arrow["{\dl}", from=7-4, to=8-4]
	\arrow["{\Sigma_\com([\id\times\inl,\id\times\inr],\id)}",from=8-2,to=9-1,bend left=15,pos=.3]
	\arrow["{T(\iota^\com\comp\Sigma_\com(\eta,\eta))}", from=8-4, to=9-4]
	\arrow["{\rho^\com}"{pos=.7}, from=9-1, to=9-4]
\end{tikzcd}
\end{equation*}
Let us rewrite $\beta^\klstar\comp\dl\comp\Sigma_\com(\hat\gamma^\com,\id)\comp\Sigma_\com(\iota,\id)\comp\theta$ equivalently
as follows:
\begin{flalign*}
&& \beta^\klstar\comp\dl\comp&\;\Sigma_\com(\hat\gamma^\com,\id)\comp\Sigma_\com(\iota,\id)\comp\theta\\*
&&  &\;=\beta^\klstar\comp\dl\comp\Sigma_\com([\eta\comp\iota^\val,\gamma^\com],\id)\comp\theta &\\
&&  &\;= \beta^\klstar\comp\dl\comp \Sigma_\com(T\nabla,\nabla)\comp\Sigma_\com(\eta\comp\iota^\val\klplus\gamma^\com,\inl)\comp\theta\\
&&  &\;= \beta^\klstar\comp T\Sigma_\com(\nabla,\nabla)\comp\dl\comp\Sigma_\com(\eta\comp\iota^\val\klplus\gamma^\com,\inl)\comp\theta\\
&&  &\;= [\eta,\beta]^\klstar\comp T\Sigma\nabla \comp T\inr\comp\dl\comp\Sigma_\com(\eta\comp\iota^\val\klplus\gamma^\com,\inl)\comp\theta\\*
&&  &\;= [\eta,\beta]^\klstar\comp T\iota^\mone\comp T\nabla^\klstar\comp T(\iota\comp\Sigma\eta\comp\inr)\comp\dl\comp\Sigma_\com(\eta\comp\iota^\val\klplus\gamma^\com,\inl)\comp\theta\\
&&  &\;= [\eta,\beta]^\klstar\comp T\iota^\mone\comp T\nabla^\klstar\comp T(\iota^\com\comp\Sigma_\com(\eta,\eta))\comp\dl\comp\Sigma_\com(\eta\comp\iota^\val\klplus\gamma^\com,\inl)\comp\theta\\
&&  &\;= [\eta,\beta]^\klstar\comp T\iota^\mone\comp T\nabla^\klstar
\intertext{Now, we can apply the above diagram and obtain the goal as follows:}
&&  &\;= [\eta,\beta]^\klstar\comp T\iota^\mone\comp T\nabla^\klstar\comp\rho^\com\comp{\Sigma_\com(\brks{\iota,\gamma^\val+\gamma^\com},\id)}\comp\theta\\
&&  &\;= [\eta,\beta]^\klstar\comp T\iota^\mone\comp\gamma^\com\comp\Sigma_\com(\iota,\id)\comp\theta\\*
&&  &\;= \beta\comp\Sigma_\com(\iota,\id)\comp\theta.
\end{flalign*}
This completes the proof of commutativity of~\eqref{eq:beta-delta}.
Let us proceed with the main goal.
For every $n\in\nats$ let $\hat\beta_n\c\mS\to T(\mSv)$ be recursively defined as follows: $\hat\beta_0 = [\eta,\bot]\comp\iota^\mone$,
$\hat\beta_{n+1} = \hat\beta_{n}^\klstar\comp\hat\gamma^\com$. Observe that $\hat\beta = \bigjoin_n\hat\beta_n$. Indeed, by definition
(Kleene's fixpoint theorem), $\beta = \bigjoin_n\beta_n$ where $\beta_0=\bot$, $\beta_{n+1} = [\eta,\beta_n]^\klstar\comp T\iota^\mone\comp\gamma^\com$, and then,
by induction, $\hat\beta_1=\hat\beta_0^\klstar\comp\gamma^\com = [\eta,\beta_0]^\klstar\comp T\iota^\mone\comp\gamma^\com = \beta_1$ and
$\hat\beta_{n+1}=\hat\beta_n^\klstar\comp\gamma^\com =(\beta_n)^\klstar\comp\gamma^\com = [\eta,\beta_n]^\klstar\comp T\iota^\mone\comp\gamma^\com = \beta_{n+1}$.
It thus suffices to prove
\begin{align*}
\hat\beta^\klstar\comp T\nabla^\klstar\comp(\rho^{\com\val})^\klstar\comp\dl\comp \Sigma_\com(T\brks{\iota^\val,\gamma^\val}\comp\hat\beta_n,\id)\appr\beta
\end{align*}
for every $n\in\nats$. We proceed by induction on $n$.

\medskip\noindent Induction base: $n = 0$. The goal is obtained as follows:
\begin{flalign*}
&&\hat\beta^\klstar\comp T\nabla^\klstar&\comp(\rho^{\com\val})^\klstar\comp\dl\comp \Sigma_\com(T\brks{\iota^\val,\gamma^\val}\comp[\eta,\bot]\comp\iota^\mone,\id)\\
&&\;=&\;\hat\beta^\klstar\comp T\nabla^\klstar\comp(\rho^{\com\val})^\klstar\comp\dl\comp \Sigma_\com([\eta\comp\brks{\iota^\val,\gamma^\val},\bot]\comp\iota^\mone,\id)\\
&&\;=&\;\hat\beta^\klstar\comp T\nabla^\klstar\comp\rho^{\com\val}\comp\Sigma_\com([\brks{\iota^\val,\gamma^\val},\bot]\comp\iota^\mone,\id)\\
&&\;=&\;\hat\beta^\klstar\comp T\nabla^\klstar\comp\rho^{\com}\comp\Sigma_\com([\brks{\iota^\val,\inl\comp\gamma^\val},\bot]\comp\iota^\mone,\id)\\
&&\;\appr&\;\hat\beta^\klstar\comp T\nabla^\klstar\comp \rho^\com\comp\Sigma_\com([\brks{\iota^\val,\inl\comp\gamma^\val},\brks{\iota^\com,\inr\comp\gamma^\com}]\comp\makebox[0pt][l]{$\iota^\mone,\id)$}\\
&&\;=&\;\hat\beta^\klstar\comp T\nabla^\klstar\comp \rho^\com\comp\Sigma_\com(\brks{\id,[\inl\comp\gamma^\val,\inr\comp\gamma^\com]\comp\iota^\mone},\id)\\
&&\;=&\;[\eta,\beta]^\klstar\comp T\iota^\mone\comp\gamma^\com\\
&&\;=&\;\beta.
\intertext{Induction step: $n > 0$. We have}
&&\hat\beta^\klstar\comp T\nabla^\klstar&\comp(\rho^{\com\val})^\klstar\comp\dl\comp \Sigma_\com(T\brks{\iota^\val,\gamma^\val}\comp\hat\beta_n,\id)\\
&&\;=&\;\hat\beta^\klstar\comp T\nabla^\klstar\comp(\rho^{\com\val})^\klstar\comp\dl\comp\Sigma_\com(T\brks{\iota^\val,\gamma^\val}\comp \hat\beta_{n-1}^\klstar\comp\hat\gamma^\com,\id)&\\
&&\;=&\;\hat\beta^\klstar\comp T\nabla^\klstar\comp(\rho^{\com\val})^\klstar\comp\dl^\klstar\comp\dl\comp\Sigma_\com(TT\brks{\iota^\val,\gamma^\val}\comp T\hat\beta_{n-1}\comp\hat\gamma^\com,\id)&\\
&&\;=&\;\hat\beta^\klstar\comp T\nabla^\klstar\comp(\rho^{\com\val})^\klstar\comp\dl^\klstar\comp T\Sigma_\com(T\brks{\iota^\val,\gamma^\val}\comp\hat\beta_{n-1},\id)\comp\dl\comp\Sigma_\com(\hat\gamma^\com,\id).&
\end{flalign*}
By induction hypothesis, the latter is smaller or equal $\beta^\klstar\comp\dl\comp\Sigma_\com(\hat\gamma^\com,\id)$,
which is~$\beta$ by \eqref{eq:beta-delta}.
\qed

\subsection{Proof of \autoref{lem:big-and-small}}
The diagram~\eqref{eq:op-mod-c} identifies
$\gamma^\com$ as the unique fixpoint of $f\mto T\nabla^\klstar\comp\rho^\com\comp\Sigma_\com(\brks{\id,(\gamma^\val+ f)\comp\iota^\mone},\id)$,
hence, also as the least one. Thus, $\gamma^\com = \bigjoin_n \gamma^\com_n$ where $\gamma^\com_0 = \bot$ and $\gamma^\com_{n+1} = T\nabla^\klstar\comp\rho^\com\comp\Sigma_\com(\brks{\id,(\gamma^\val+\gamma_n^\com)\comp\iota^\mone},\id)$, and it suffices to show 
\begin{align}\label{eq:main-ind}
	\hat\zeta^\klstar\comp\gamma^\com_n\appr\zeta
\end{align}
for all $n$. The induction base is obvious. For the induction step, assume~\eqref{eq:main-ind}
and show $\hat\zeta^\klstar\comp\gamma^\com_{n+1}\appr\zeta$. By \autoref{ass:compl}, $\mSc$ is a coproduct of $\Sigma_\com(\mSv,\mS)$ and $\Theta(\mSv,\mSc,\mS)$.
We thus reduce to two inequations:
\begin{align*}
\hat\zeta^\klstar\comp\gamma^\com_{n+1}\comp\Sigma_\com(\iota^\val,\id) 		\appr&\; \zeta\comp\Sigma_\com(\iota^\val,\id),\\*
\hat\zeta^\klstar\comp\gamma^\com_{n+1}\comp\Sigma_\com(\iota,\id)\comp\theta   \appr&\; \zeta\comp\Sigma_\com(\iota,\id)\comp\theta.
\end{align*}
The first inequation is easy to obtain by unfolding definitions:
\begin{flalign*}
&& \hat\zeta^\klstar\comp\gamma^\com_{n+1}&\,\comp\Sigma_\com(\iota^\val,\id)\\*
&&  &\;= \hat\zeta^\klstar\comp T\nabla^\klstar\comp\rho^\com\comp{\Sigma_\com(\brks{\iota,\gamma^\val+\gamma^\com_n},\id)}\comp\Sigma_\com(\inl,\id)&\\
&&  &\;= \hat\zeta^\klstar\comp T\nabla^\klstar\comp\rho^\com\comp\Sigma_\com(\brks{\iota^\val,\inl\comp\gamma^\val},\id)\\
&&  &\;= \hat\zeta^\klstar\comp T\nabla^\klstar\comp\rho^{\com\val}\comp\Sigma_\com(\brks{\iota^\val,\gamma^\val},\id)  \\
&&  &\;= \hat\zeta^\klstar\comp T\nabla^\klstar\comp(\rho^{\com\val})^\klstar\comp\dl\comp \Sigma_\com(T\brks{\iota^\val,\gamma^\val}\comp\eta,\id)  \\
&&  &\;= \hat\zeta^\klstar\comp T\nabla^\klstar\comp(\rho^{\com\val})^\klstar\comp\dl\comp \Sigma_\com(T\brks{\iota^\val,\gamma^\val}\comp \hat\zeta,\id)\comp\Sigma_\com(\iota^\val,\id)  \\
&&  &\;= \zeta\comp\Sigma_\com(\iota^\val,\id). &\by{\autoref{lem:bs-fix}}
\end{flalign*}
For the second inequation, we need the following auxiliary property:
\begin{align}\label{eq:main-help}
\Sigma_\com(\eta\comp\fst\klplus\snd,\inl)\comp\theta\comp\Theta(\brks{\iota^\val,\gamma^\val},\brks{\iota^\com,\gamma^\com_n},\id)
= \Sigma_\com(\eta\comp\iota^\val\klplus\gamma^\com_n,\inl)\comp\theta,
\end{align}
which entails the goal as follows:
\begin{flalign*}
&& \hat\zeta^\klstar\comp\gamma&\,^\com_{n+1}\comp\Sigma_\com(\iota,\id)\comp\theta\\
&&  &\;= \hat\zeta^\klstar\comp T\nabla^\klstar\comp\rho^\com\comp{\Sigma_\com(\brks{\iota,\gamma^\val+\gamma^\com_n},\id)}\comp\theta&\\
&&  &\;= \hat\zeta^\klstar\comp T\nabla^\klstar\comp\rho^\com\comp{\Sigma_\com([\id\times\inl,\id\times \inr],\id)}\comp\theta\comp\Theta(\brks{\iota^\val,\gamma^\val},\brks{\iota^\com,\gamma^\com_n},\id)\\
&&  &\;= \hat\zeta^\klstar\comp T\nabla^\klstar\comp{T(\iota^\com\comp\Sigma_\com(\eta,\eta))}\comp\dl\comp\\
	&&&\qquad\qquad\Sigma_\com(\eta\comp\fst\klplus\snd,\inl)\comp\theta\comp\Theta(\brks{\iota^\val,\gamma^\val},\brks{\iota^\com,\gamma^\com_n},\id)&\by{\eqref{eq:str_sep}}\\
&&  &\;= \hat\zeta^\klstar\comp T\nabla^\klstar\comp T\iota^\com\comp T\Sigma_\com(\eta,\eta)\comp\dl\comp\Sigma_\com(\eta\comp\iota^\val\klplus\gamma^\com_n,\inl)\comp\theta&\by{\eqref{eq:main-help}}\\
&&  &\;= \zeta^\klstar\comp T\Sigma_\com(\nabla^\klstar,\nabla^\klstar)\comp T\Sigma_\com(\eta,\eta)\comp\dl\comp\Sigma_\com(\eta\comp\iota^\val\klplus\gamma^\com_n,\inl)\comp\theta\\
&&  &\;= \zeta^\klstar\comp T\Sigma_\com(\nabla,\nabla)\comp\dl\comp\Sigma_\com(\eta\comp\iota^\val\klplus\gamma^\com_n,\inl)\comp\theta\\
&&  &\;= \zeta^\klstar\comp\dl\comp \Sigma_\com(T\nabla,\nabla)\comp\Sigma_\com(\eta\comp\iota^\val\klplus\gamma^\com_n,\inl)\comp\theta\\
&&  &\;= \zeta^\klstar\comp\dl\comp\Sigma_\com([\eta\comp\iota^\val,\gamma^\com_n],\id)\comp\theta\\
&&  &\;= \hat\zeta^\klstar\comp T\nabla^\klstar\comp(\rho^{\com\val})^\klstar\comp\dl^\klstar\comp T\Sigma_\com(T\brks{\iota^\val,\gamma^\val}\comp \hat\zeta,\id)\comp\dl\comp\Sigma_\com([\eta\comp\iota^\val,\gamma^\com_n],\id)\comp\theta\\
&&  &\;= \hat\zeta^\klstar\comp T\nabla^\klstar\comp(\rho^{\com\val})^\klstar\comp\dl^\klstar\comp\dl\comp \Sigma_\com(TT\brks{\iota^\val,\gamma^\val}\comp T\hat\zeta,\id)\comp\Sigma_\com([\eta\comp\iota^\val,\gamma^\com_n],\id)\comp\theta\\
&&  &\;= \hat\zeta^\klstar\comp T\nabla^\klstar\comp(\rho^{\com\val})^\klstar\comp\dl\comp\Sigma_\com(T\brks{\iota^\val,\gamma^\val}\comp \hat\zeta^\klstar\comp[\eta\comp\iota^\val,\gamma^\com_n],\id)\comp\theta\\
&&  &\;\appr \hat\zeta^\klstar\comp T\nabla^\klstar\comp(\rho^{\com\val})^\klstar\comp\dl\comp\Sigma_\com(T\brks{\iota^\val,\gamma^\val}\comp \hat\zeta^\klstar\comp[\eta\comp\iota^\val,\eta\comp\iota^\com],\id)\comp\theta&\by{\eqref{eq:main-ind}}\\
&&  &\;= \zeta^\klstar\comp\dl\comp\Sigma_\com([\eta\comp\iota^\val,\eta\comp\iota^\com],\id)\comp\theta\\
&&  &\;= \zeta\comp\Sigma_\com(\iota,\id)\comp\theta.
\intertext{The equation~\eqref{eq:main-help} is shown as follows:}
&& \!\!\!\Sigma_\com(\eta\comp\fst&\;\klplus\snd,\inl)\comp\theta\comp\Theta(\brks{\iota^\val,\gamma^\val},\brks{\iota^\com,\gamma^\com_n},\id)\\*
&&  &\;= \Sigma_\com(\id\klplus\id,\inl)\comp\Sigma_\com(\eta\comp\fst+\snd,\id)\comp\theta\comp\Theta(\brks{\iota^\val,\gamma^\val},\brks{\iota^\com,\gamma^\com_n},\id)\\
&&  &\;= \Sigma_\com(\id\klplus\id,\inl)\comp\theta\comp\Theta(\eta\comp\fst,\snd,\id)\comp\Theta(\brks{\iota^\val,\gamma^\val},\brks{\iota^\com,\gamma^\com_n},\id)\\
&&  &\;= \Sigma_\com(\id\klplus\id,\inl)\comp\theta\comp\Theta(\eta,\id,\id)\comp\Theta(\iota^\val,\gamma^\com_n,\id)\\
&&  &\;= \Sigma_\com(\id\klplus\id,\inl)\comp\Sigma_\com(\eta+\id,\id)\comp\theta\comp\Theta(\iota^\val,\gamma^\com_n,\id)\\
&&  &\;= \Sigma_\com(\eta\klplus\id,\inl)\comp\theta\comp\Theta(\iota^\val,\gamma^\com_n,\id)\\
&&  &\;= \Sigma_\com(\eta\comp\iota^\val\klplus\gamma^\com_n,\inl)\comp\theta&\tag*{$\qed$}
\end{flalign*}

\subsection{Proof of \autoref{lem:bs-fix}}
By rewriting the relevant subexpression of~\eqref{eq:bs-xi}:
\begin{flalign*}
&& [\eta,&\,f^\klstar\comp T\mu\comp\xi^\klstar\comp\dl\comp\Sigma_\com(f,\id)]\comp\iota^\mone&\\
&&  &\;= [\eta, f^\klstar\comp T\nabla^\klstar\comp(\rho^{\com\val})^\klstar\comp T\Sigma_\com(\brks{\iota^\val,\gamma^\val},\id)\comp\dl\comp\Sigma_\com(f,\id)]\comp\iota^\mone&\by{\autoref{lem:xi-mS}}\\
&&  &\;= [\eta, f^\klstar\comp T\nabla^\klstar\comp(\rho^{\com\val})^\klstar\comp\dl\comp \Sigma_\com(T\brks{\iota^\val,\gamma^\val},\id)\comp\Sigma_\com(f,\id)]\comp\iota^\mone\\
&&  &\;= [\eta, f^\klstar\comp T\nabla^\klstar\comp(\rho^{\com\val})^\klstar\comp\dl\comp \Sigma_\com(T\brks{\iota^\val,\gamma^\val}\comp f,\id)]\comp\iota^\mone. &\tag*{$\qed$}
\end{flalign*}

\subsection{Proof of \autoref{lem:bs-fix-lam}}
Observe first, that the following diagram commutes:
\begin{equation*}
\begin{tikzcd}[column sep=-7em, row sep=0ex]
\Sigma_\com(\Sigma_\val (\mS\times(\mS\monto\mS)),\mS)
	\dar["{\Sigma_\com(\Sigma_\val(\eta\times\id),\eta)}"']
    \ar[rr,"{\Sigma_\com(\brks{\iota^\val\comp\Sigma_\val\fst,\,\rho^\val},\id)}"]
&[2.5em] &[.5em]
\Sigma_\com(\mS\times D(\mS,\Sigmas(\mS+\mS)),\mS)
	\ar[ddddddd, "{\Sigma_\com(\id\times D(\id,\nabla^\klstar),\id)}"]
\\[4ex]
\Sigma_\com(\Sigma_\val (\mu^2\Sigma\times(\mS\monto\mS)),\mu^2\Sigma)
	\dar["{\Sigma_\com(\Sigma_\val (\id\times(\mmul\monto\id)),\id)}"']
	\ar[ddr,"{\Sigma_\com(\brks{\iota^\val\comp\Sigma_\val\fst,\Sigma_\val(\mmul\times\id)},\id)}", bend left=20, pos=.3]
& &
\\[4ex]
\Sigma_\com(\Sigma_\val(\mu^2\Sigma\times(\mu^2\Sigma\monto\mS)),\mu^2\Sigma)
    \ar[ddd,"{\Sigma_\com(\brks{\iota^\val\comp\Sigma_\val\fst,\, \rho^\val},\id)}"']
& &[3.25em]
\\[2ex]
&
\Sigma_\com(\mu^2\Sigma\times\Sigma_\val(\mS\times (\mS\monto\mS)),\mu^2\Sigma)\qquad\quad
	\ar[d, "{\Sigma_\com(\id\times \rho^\val,\id)}"']
&
\\[5ex]
&
\Sigma_\com(\mu^2\Sigma\times D(\mS,\Sigmas(\mS+\mS)),\mu^2\Sigma)\quad
	\ar[dd,"{\Sigma_\com(\id\times D(\mmul,\id),\id)}",pos=.3]
	\ar[ruuuu, "{\Sigma_\com(\mmul\times\id,\mmul)}", bend right=35, pos=.6, end anchor=-120]
&
\\[2ex]
\Sigma_\com(\mu^2\Sigma\times D(\mu^2\Sigma,\Sigmas(\mu^2\Sigma+\mS)),\mu^2\Sigma)
	\ar[dd,"\rho^{\com\val}"']
	\ar[dddr, "{\Sigma_\com(\id\times D(\id,[\mu,\id]^\klstar),\id)}", bend right=25,pos=.475]
	\ar[dr, "{\Sigma_\com(\id\times D(\id,\mu+\id),\id)}", bend right=15,pos=.55]
& 
&
\\[10ex]
& 
\Sigma_\com(\mu^2\Sigma\times D(\mu^2\Sigma,\Sigmas(\mS+\mS)),\mu^2\Sigma)
	\ar[dd, "{\Sigma_\com(\id\times D(\id,\nabla^\klstar),\id)}",pos=.75]
&
\\[8ex]
T\Sigmas(\mu^2\Sigma+\Sigmas(\mu^2\Sigma+\mS))\qquad
	\ar[dddd,"{T[\Sigmas\inl,[\Sigmas\inl,\eta\comp\inr]^\klstar]^\klstar}"']
	\ar[dddr, "{T\Sigmas(\id+[\mu,\id]^\klstar)}", bend right=30,pos=.2]
& &
\Sigma_\com(\mS\times D(\mS,\mS),\mS)
	\ar[dd,"{\rho^{\com\val}}"]
\\[8ex]
& 
\Sigma_\com(\mu^2\Sigma\times D(\mu^2\Sigma,\mS),\mu^2\Sigma)\quad
	\ar[dd,"\rho^{\com\val}",pos=.3]
&
\\[1ex]
{}
& 
&
T\Sigmas(\mS+\mS)
	\ar[dd,"T\nabla^\klstar"]
\\[5ex]
&
T\Sigmas(\mu^2\Sigma+\mS) 
	\ar[ur,"T\Sigmas(\mmul+\id)", bend right=15]&
&
\\[2ex]
T\Sigmas(\mS+\mS)
	\ar[rr,"T\nabla^\klstar"]
& &
T(\mS)
\end{tikzcd}
\end{equation*}
We then prove the claim by modifying the argument from \autoref{lem:bs-fix}.
It suffices to show that
\begin{align*}
T\nabla^\klstar\comp\xi\comp\Sigma_\com(\Sigma_\val\brks{\id,(\oname{subst}\comp(\id\times\mu))^{\flat}},\id) 
= T\nabla^\klstar\comp\rho^{\com\val}\comp\Sigma_\com(\brks{\iota^\val,\gamma^\val},\id).
\end{align*}
Using the definition 
\begin{align*}
\xi=T[\Sigmas\inl,[\Sigmas\inl,\eta\comp\inr]^\klstar]^\klstar\comp\rho^{\com\val}\comp\Sigma_\com(\brks{\iota^\val\comp\Sigma_\val\fst,\,   \rho^\val},\id)\comp\Sigma_\com(\Sigma_\val(\eta\times\id),\eta)
\end{align*}
and the above diagram,
\begin{align*}
T\nabla^\klstar\comp\xi\comp&\,\Sigma_\com(\Sigma_\val\brks{\id,(\oname{subst}\comp(\id\times\mu))^{\flat}},\id) \\*
=\;& T\nabla^\klstar\comp T[\Sigmas\inl,[\Sigmas\inl,\eta\comp\inr]^\klstar]^\klstar\comp\rho^{\com\val}\comp\Sigma_\com(\brks{\iota^\val\comp\Sigma_\val\fst,\, \rho^\val},\id)\comp\\
&\qquad\Sigma_\com(\Sigma_\val(\eta\times\id),\eta) \comp\Sigma_\com(\Sigma_\val\brks{\id,(\oname{subst}\comp(\id\times\mu))^{\flat}},\id)\\
=\;& T\nabla^\klstar\comp T[\Sigmas\inl,[\Sigmas\inl,\eta\comp\inr]^\klstar]^\klstar\comp\rho^{\com\val}\comp\Sigma_\com(\brks{\iota^\val\comp\Sigma_\val\fst,\, \rho^\val},\id)\comp\\
&\qquad\Sigma_\com(\Sigma_\val(\eta\times\id),\eta) \comp\Sigma_\com(\Sigma_\val (\id\times(\mmul\monto\id)),\id)\comp\Sigma_\com(\Sigma_\val\brks{\id,\oname{subst}^{\flat}},\id) \\
=\;& T\nabla^\klstar\comp T[\Sigmas\inl,[\Sigmas\inl,\eta\comp\inr]^\klstar]^\klstar\comp\rho^{\com\val}\comp\Sigma_\com(\brks{\iota^\val\comp\Sigma_\val\fst,\, \rho^\val},\id)\comp\\
&\qquad\Sigma_\com(\Sigma_\val (\id\times(\mmul\monto\id)),\id)\comp\Sigma_\com(\Sigma_\val(\eta\times\id),\eta) \comp\Sigma_\com(\Sigma_\val\brks{\id,\oname{subst}^{\flat}},\id) \\
=\;& T\nabla^\klstar\comp\rho^{\com\val}\comp \Sigma_\com(\id\times D(\id,\nabla^\klstar),\id)\comp\\
&\qquad\Sigma_\com(\brks{\iota^\val\comp\Sigma_\val\fst,\, \rho^\val},\id) \comp\Sigma_\com(\Sigma_\val\brks{\id,\oname{subst}^{\flat}},\id)\\
=\;& T\nabla^\klstar\comp\rho^{\com\val}\comp\Sigma_\com(\brks{\iota^\val\comp\Sigma_\val\fst,\, D(\id,\nabla^\klstar)\comp\rho^\val},\id) \comp\Sigma_\com(\Sigma_\val\brks{\id,\oname{subst}^{\flat}},\id)\\
=\;& T\nabla^\klstar\comp\rho^{\com\val}\comp\Sigma_\com(\brks{\iota^\val,\, D(\id,\nabla^\klstar)\comp\rho^\val\comp \Sigma_\val\brks{\id,\oname{subst}^{\flat}}},\id)\\
=\;& T\nabla^\klstar\comp\rho^{\com\val}\comp\Sigma_\com(\brks{\iota^\val,\gamma^\val},\id),     
\end{align*}
as desired.\qed

\end{document}